\documentclass[journal]{IEEEtran}
\usepackage{amsmath}
\usepackage{amssymb}
\usepackage{amsfonts}
\usepackage{graphicx}
\usepackage{epsfig}
\usepackage{subfigure}
\usepackage{psfrag}
\usepackage{color}

\begin{document}

\title{Optimal Beamforming for Two-Way Multi-Antenna Relay Channel with Analogue Network Coding}

\author{Rui Zhang,~\IEEEmembership{Member,~IEEE,}
        Ying-Chang Liang,~\IEEEmembership{Senior Member,~IEEE,}
        Chin Choy Chai,~\IEEEmembership{Member,~IEEE}, and Shuguang Cui, ~\IEEEmembership{Member,~IEEE} 
\thanks{Manuscript received August 1, 2008, revised February 16, 2009.}
\thanks{R. Zhang, Y.-C. Liang, and C. C. Chai are with the Institute for Infocomm Research, A*STAR, Singapore. (e-mails:
\{rzhang, ycliang, chaicc\}@i2r.a-star.edu.sg)}
\thanks{S. Cui is with the Department of Electrical and Computer
Engineering, Texas A\&M University, Texas, USA. (e-mail:
cui@ece.tamu.edu)}}

\maketitle

\begin{abstract}
This paper studies the wireless \emph{two-way relay channel} (TWRC),
where two source nodes, S1 and S2, exchange information through an
assisting relay node, R. It is assumed that R receives the sum
signal from S1 and S2 in one time-slot, and then amplifies and
forwards the received signal to both S1 and S2 in the next
time-slot. By applying the principle of \emph{ analogue network
coding} (ANC), each of S1 and S2 cancels the so-called
``self-interference'' in the received signal from R and then decodes
the desired message. Assuming that S1 and S2 are each equipped with
a single antenna and R with multi-antennas, this paper analyzes the
\emph{capacity region} of the ANC-based TWRC with linear processing
(beamforming) at R. The capacity region contains all the achievable
bidirectional rate-pairs of S1 and S2 under the given transmit power
constraints at S1, S2, and R. We present the optimal relay
beamforming structure as well as an efficient algorithm to compute
the optimal beamforming matrix based on convex optimization
techniques. Low-complexity suboptimal relay beamforming schemes are
also presented, and their achievable rates are compared against the
capacity with the optimal scheme.
\end{abstract}

\begin{keywords}
Analogue network coding, beamforming, convex optimization, two-way
relay channel.
\end{keywords}

\IEEEpeerreviewmaketitle

\setlength{\baselineskip}{1.0\baselineskip}
\newtheorem{theorem}{\underline{Theorem}}[section]
\newtheorem{lemma}{\underline{Lemma}}[section]
\newtheorem{remark}{\underline{Remark}}[section]
\newtheorem{corollary}{\underline{Corollary}}[section]
\newtheorem{problem}{\underline{Problem}}[section]
\newtheorem{algorithm}{\underline{Algorithm}}[section]
\newcommand{\mv}[1]{\mbox{\boldmath{$ #1 $}}}

\section{Introduction}

\PARstart{N}etwork coding \cite{Network_coding} is a new and
promising design paradigm for modern communication networks: By
allowing intermediate network nodes to mix the data or signals
received from multiple links, as opposed to separating them by
traditional approaches, network coding reduces the amount of
transmissions in the network and thus improves the overall network
throughput. Recently, there has been increasing attention from the
research community to apply the principle of network coding in
wireless communication networks. In fact, wireless network is the
most natural setting to apply network coding due to the broadcast
property of radio transmissions, i.e., a single transmission of one
wireless terminal may successfully reach multiple neighboring
terminals, without the need of dedicated links to these terminals as
required in wireline networks. Furthermore, network coding can
potentially be a very effective solution to the classical
``interference problem'' in wireless networks, since it transforms
the traditional approach of avoiding or mitigating the interference
among wireless terminals into a new methodology of {\it interference
exploitation}.

The {\it two-way relay channel} (TWRC) is one of the basic elements
in decentralized/centralized wireless networks. The simplest TWRC
consists of two source nodes, S1 and S2, which exchange information
via a helping relay node, R. Traditionally, in order to avoid the
interference at R, simultaneous transmission of S1 and S2 is
unadvisable at the same frequency. Thus, in total four time-slots
are usually required to accomplish one round of information exchange
between S1 and S2 via R. However, by applying the idea of network
coding, the authors in \cite{Wu05} proposed a method to reduce the
number of required time-slots from four to three. In this method, S1
first sends to R during time-slot 1 the message $s_1$ consisting of
bits $b_1(1),\ldots,b_1(N)$ with $N$ denoting the message length in
bits, and then R decodes $s_1$. During time-slot 2, S2 sends to R
the message $s_2$ consisting of bits $b_{2}(1),\ldots, b_{2}(N)$,
and R decodes $s_{2}$. In time-slot 3, R broadcasts to S1 and S2 a
new message $s_3$ consisting of bits $b_{3}(1),\ldots, b_{3}(N)$
obtained by bit-wise exclusive-or (XOR) operations over $b_1(n)$'s
and $b_2(n)$'s, i.e., $b_3(n)=b_1(n) \oplus b_2(n), \forall n$.
Since S1 knows $b_1(n)$'s, S1 can recover its desired message $s_2$
by first decoding $s_3$ and then obtaining $b_2(n)$'s as $b_{1}(n)
\oplus b_{3}(n), \forall n$. Similarly, S2 can recover $s_1$.

The principle of network coding has been further investigated for
TWRC by exploiting various physical-layer relay operations
\cite{ANC}, \cite{PHY_network_coding}. The scheme proposed in
\cite{ANC} is named as {\it analogue network coding} (ANC), while
the one in \cite{PHY_network_coding} named as {\it physical-layer
network coding} (PNC). For both ANC and PNC, the number of
time-slots required for S1 and S2 to exchange one round of
information is reduced from three \cite{Wu05} to two, by allowing S1
and S2 to transmit simultaneously to R during one time-slot and
thereby combining the first two time-slots in \cite{Wu05} into one
time-slot. ANC and PNC differ in their corresponding relay
operations, which are amplify-and-forward (AF) and
estimate-and-forward (EF), respectively. In ANC, R linearly
amplifies the sum signal received from S1 and S2, and then
broadcasts the resulting signal to S1 and S2. ANC is based upon an
interesting observation that the signal collision at R during the
first time-slot is in fact harmless, since such a collision can be
resolved at S1 (S2) during the second time-slot by subtracting from
its received signal the so-called {\it self-interference}, which is
related to the previously transmitted message from S1 (S2) itself.
In contrast to ANC, more sophisticated (nonlinear) operations than
AF are required at R for PNC \cite{PHY_network_coding}-\cite{Zhang}.
Instead of decoding messages $s_1$ from S1 and $s_2$ from S2
separately in two different time-slots like in \cite{Wu05}, the EF
method proposed in \cite{PHY_network_coding} estimates at R the
bitwise XORs between $b_{1}(n)$'s and $b_{2}(n)$'s from the mixed
signal of S1 and S2, and re-encodes the decoded bits into a new
broadcasting message $s_3$; each one of S1 and S2 then recovers the
other's message by the same decoding method as that in \cite{Wu05}.
Alternatively, it is possible to first deploy multiuser decoding at
R to decode $s_1$ and $s_2$ separately, and then jointly encode
$s_1$ and $s_2$ into a new broadcasting message $s_3$; given the
side information on $s_1$ ($s_2$) at S1 (S2), S1 (S2) decodes $s_2$
($s_1$). The above decode-and-forward (DF) relay operation for TWRC
has been studied in \cite{Boche08a}, \cite{Xie07}. On the other
hand, TWRC has also been studied in \cite{Rankov05}-\cite{Tarokh07}
from cooperative communication perspectives, with a major objective
to compensate for the loss of spectral efficiency in the
conventional {\it one-way relay channel} (OWRC) owing to the
half-duplex constraint. Non-surprisingly, the solutions proposed
therein are similar to those inspired by the principle of network
coding.

Furthermore, TWRC has been studied jointly with other physical-layer
transmission techniques based on, e.g.,
orthogonal-frequency-division-multiplexing (OFDM) \cite{Ho98},
\cite{two_way_OFDMA}, and multiple transmit and/or multiple receive
antennas \cite{Ham07}-\cite{Cuitao}, to further improve the
bidirectional relay throughput. For the multi-antenna TWRC, the DF
relay strategy was studied in \cite{Ham07}, \cite{Boche08b}, the AF
relay strategy or ANC was studied in \cite{Unger2007}, and the
distributed space-time coding strategy for the relay was studied in
\cite{Cuitao}. In this paper, we focus on the AF-/ANC-based
multi-antenna TWRC. Assuming that S1 and S2 each has a single
antenna and R has $M$ antennas, $M\geq 2$, we study the optimal
design of linear processing (beamforming) at the relay to achieve
the capacity region of AF-/ANC-based TWRC, which consists of all the
achievable rate-pairs of S1 and S2 under the given transmit power
constraints at S1, S2, and R. Our main goal is to provide insightful
guidelines on the design of AF-based multi-antenna TWRC, which
differs from the results for the conventional AF-based multi-antenna
OWRC given in, e.g., \cite{Hua07}-\cite{Zhang08}. The main results
of this paper are summarized as follows:\footnote{Preliminary
versions of this paper have been presented in \cite{Liang},
\cite{Rui}.}
\begin{itemize}

\item We derive the optimal beamforming structure at R, which
achieves the capacity region of an ANC-based TWRC. The optimal
structure reduces the number of complex-valued design variables in
the relay beamforming matrix from $M^2$ to 4 when $M>2$.
Furthermore, by transforming the capacity region characterization
problem into an equivalent relay power minimization problem under
certain signal-to-noise-ratio (SNR) constraints at S1 and S2, we
derive an efficient algorithm to compute the globally optimal
beamforming matrix based on convex optimization techniques.

\item Inspired by the optimal relay beamforming
structure, we propose two low-complexity suboptimal beamforming
schemes, based on the principle of ``matched-filter (MF)'' and
``zero-forcing (ZF)'', respectively. We analyze their performances
in terms of the achievable sum-rate in TWRC against the maximum
sum-rate, or the sum-capacity, achieved by the optimal scheme. It is
shown that the ZF-based relay beamforming with the objective of
suppressing the uplink (from S1 and S2 to R) and downlink (from R to
S1 and S2) interferences at R may not be a good solution for the
ANC-based TWRC, since these interferences are indeed
self-interferences and thus can be later removed at S1 and S2. On
the other hand, it is shown that the MF-based relay beamforming,
which maximizes the signal power forwarded to S1 and S2, achieves
the sum-rate close to the sum-capacity under various SNR and channel
conditions.
\end{itemize}

The rest of this paper is organized as follows. Section
\ref{sec:system model} describes the TWRC model with ANC. Section
\ref{sec:capacity region} studies the capacity region of the
ANC-based TWRC, derives the optimal structure for relay beamforming,
and proposes an algorithm to compute the optimal beamforming matrix.
Section \ref{sec:suboptimal schemes} presents the low-complexity
suboptimal relay beamforming schemes. Section \ref{sec:performances}
analyzes the performances of both the optimal and suboptimal relay
beamforming schemes in terms of the achievable sum-rate in TWRC.
Section \ref{sec:numerical results} shows numerical results on the
performances of the proposed schemes, in comparison with other
existing schemes in the literature. Finally, Section
\ref{sec:conclusions} concludes the paper.

{\it Notation}: Scalars are denoted by lower-case letters, e.g.,
$x$, and bold-face lower-case letters are used for vectors, e.g.,
$\mv{x}$, and bold-face upper-case letters for matrices, e.g.,
$\mv{X}$. In addition, $\mathtt{tr}(\mv{S})$, $|\mv{S}|$,
$\mv{S}^{-1}$, and $\mv{S}^{\frac{1}{2}}$ denote the trace,
determinant, inverse, and square-root of a square matrix $\mv{S}$,
respectively, and $\mathtt{diag}(\mv{S}_1,\ldots,\mv{S}_M)$ denotes
a block-diagonal square matrix with $\mv{S}_1,\ldots,\mv{S}_M$ as
the diagonal square matrices. $\mv{S}\succeq 0$ means that $\mv{S}$
is a positive semi-definite matrix \cite{Boydbook}. For an
arbitrary-size matrix $\mv{M}$, $\mv{M}^{T}$, $\mv{M}^*$,
$\mv{M}^{H}$, and $\mv{M}^{\dag}$ denote the transpose, conjugate,
conjugate transpose, and pseudo inverse of $\mv{M}$, respectively,
$\mv{M}(i,j)$ denotes the $(i,j)$-th element of $\mv{M}$, and
$\mathtt{rank}(\mv{M})$ denotes the rank of $\mv{M}$. $\mv{I}$ and
$\mv{0}$ denote the identity matrix and the all-zero matrix,
respectively. $\|\mv{x}\|$ denotes the Euclidean norm of a complex
vector $\mv{x}$, while $|z|$ denotes the norm of a complex number
$z$. $\mathbb{C}^{x \times y}$ denotes the space of $x\times y$
matrices with complex-valued elements. The distribution of a
circular symmetric complex Gaussian (CSCG) random vector with mean
$\mv{x}$ and covariance matrix $\mv{\Sigma}$ is denoted by
$\mathcal{CN}(\mv{x},\mv{\Sigma})$, and $\sim$ stands for
``distributed as''.

\section{System Model} \label{sec:system model}

As shown in Fig. \ref{fig:system model}, we consider a TWRC
consisting of two source nodes, S1 and S2, each with a single
antenna and a relay node, R, equipped with $M$ antennas, $M \geq 2$.
All the channels involved are assumed to be flat-fading over a
common narrow-band. It is assumed that the transmission protocol of
TWRC uses two consecutive equal-duration time-slots for one round of
information exchange between S1 and S2 via R. During the first
time-slot, both S1 and S2 transmit concurrently to R, which linearly
processes the received signal and then broadcasts the resulting
signal to S1 and S2 during the second time-slot. It is also assumed
that perfect synchronization has been established among S1, S2, and
R prior to data transmission. The received baseband signal at R in
the first time-slot is expressed as
\begin{eqnarray}\label{eq:Rx signal relay}
\mv{y}_R(n)=\mv{h}_1\sqrt{p_1}s_1(n)+\mv{h}_2\sqrt{p_2}s_2(n)+\mv{z}_R(n)
\end{eqnarray}
where $\mv{y}_R(n)\in\mathbb{C}^{M\times 1}$ is the received signal
vector at symbol index $n$, $n=1,\ldots,N$, with $N$ denoting the
total number of transmitted symbols during one time-slot;
$\mv{h}_1\in\mathbb{C}^{M\times 1}$ and
$\mv{h}_2\in\mathbb{C}^{M\times 1}$ represent the channel vectors
from S1 to R and from S2 to R, respectively, which are assumed to be
constant during the two time-slots; and $s_1(n)$ and $s_2(n)$ are
the transmitted symbols from S1 and S2, respectively. Since in this
paper we are interested in the information-theoretic limits of TWRC,
it is assumed that the optimal Gaussian codebook is used at S1 and
S2, and thus $s_1(n)$ and $s_2(n)$ are independent random variables
both $\sim\mathcal{CN}(0,1)$; $p_1$ and $p_2$ denote the transmit
powers of S1 and S2, respectively; and
$\mv{z}_R(n)\in\mathbb{C}^{M\times 1}$ is the receiver noise vector,
independent over $n$, and without loss of generality (w.l.o.g.), it
is assumed that $\mv{z}_R(n)\sim\mathcal{CN}(\mv{0},\mv{I}), \forall
n$. Upon receiving the mixed signal from S1 and S2, R processes it
with AF relay operation, also known as {\it linear analogue
relaying}, and then broadcasts the processed signal to S1 and S2
during the second time-slot. Mathematically, the linear processing
(beamforming) operation at the relay can be concisely represented as
\begin{eqnarray}\label{eq:Tx signal relay}
\mv{x}_R(n)=\mv{A}\mv{y}_R(n), \ n=1,\ldots,N
\end{eqnarray}
where $\mv{x}_R(n)\in\mathbb{C}^{M\times 1}$ is the transmitted
signal at R, and $\mv{A}\in\mathbb{C}^{M\times M}$ is the relay
processing matrix.

\begin{figure}
\psfrag{a}{S1}\psfrag{b}{R}\psfrag{c}{S2}\psfrag{d}{$\mv{h}_1$}\psfrag{e}{$\mv{h}_2$}
\psfrag{f}{$\mv{h}_1^T$}\psfrag{g}{$\mv{h}_2^T$}
\begin{center}
\scalebox{0.7}{\includegraphics*[50pt,521pt][365pt,745pt]{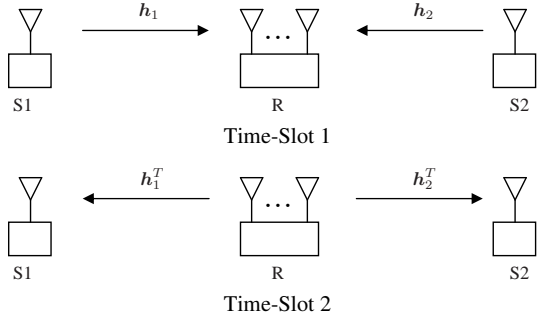}}\vspace{-0.3in}
\caption{The two-way multi-antenna relay channel.}\label{fig:system
model}\vspace{-0.2in}
\end{center}
\end{figure}

Note that the transmit power of R can be shown equal to
\begin{eqnarray}\label{eq:relay power}
p_R(\mv{A})&=&\mathtt{tr} \left(\mv{x}_R(n)\mv{x}_R^H(n)\right) \nonumber \\
&=& \| \mv{A}\mv{h}_1 \|^{2} p_{1} + \| \mv{A}\mv{h}_2 \|^{2} p_2 +
\mathtt{tr} (\mv{A}\mv{A}^{H}).
\end{eqnarray}
We can assume w.l.o.g. that channel reciprocity holds for TWRC
during uplink and downlink transmissions, i.e., the channels from R
to S1 and S2 during the second time-slot are given as $\mv{h}_1^T$
and $\mv{h}_2^T$, respectively.\footnote{This assumption is made
merely for the purpose of exposition, and the results developed in
this paper hold similarly for the more general case with independent
uplink and downlink channels.} Thus, the received signals at S1 can
be written as
\begin{eqnarray}
y_1 (n) &=& \mv{h}_1^T \mv{x}_R(n) + z_1(n) \nonumber \\
 &=& \mv{h}_1^T\mv{A}\mv{h}_1\sqrt{p_1}s_1(n)+
 \mv{h}_1^T\mv{A}\mv{h}_2\sqrt{p_2}s_2(n) \nonumber \\ && +
 \mv{h}_1^T\mv{A}\mv{z}_R(n)+z_1(n) \label{eq:Rx signal S1 2}
\end{eqnarray}
for $n=1,\ldots,N$, where $z_1(n)$'s are the independent receiver
noise samples at S1, and it is assumed that
$z_1(n)\sim\mathcal{CN}(0,1), \forall n$. Note that on the
right-hand side (RHS) of (\ref{eq:Rx signal S1 2}), the first term
is the self-interference of S1, while the second term contains the
desired message from S2. Assuming that both
$\mv{h}_1^T\mv{A}\mv{h}_1$ and $\mv{h}_1^T\mv{A}\mv{h}_2$ are
perfectly known at S1 via training-based channel estimation
\cite{Gao} prior to data transmission, S1 can first subtract its
self-interference from $y_1 (n)$ and then coherently demodulate
$s_2(n)$. The above practice is known as {\it analogue network
coding} (ANC) \cite{ANC}. From (\ref{eq:Rx signal S1 2}),
subtracting the self-interference from $y_1 (n)$ yields
\begin{eqnarray}\label{eq:Rx signal S1 3}
\tilde{y}_1(n)= \tilde{h}_{21}\sqrt{p_2}s_2(n)+ \tilde{z}_1(n), \
n=1,\ldots,N
\end{eqnarray}
where $\tilde{h}_{21}=\mv{h}_1^T\mv{A}\mv{h}_2$, and
$\tilde{z}_1(n)\sim\mathcal{CN}(0,\|\mv{A}^H\mv{h}_1^*\|^2+1)$. From
(\ref{eq:Rx signal S1 3}), for a given $\mv{A}$, the maximum
achievable rate (in bits/complex dimension) for the end-to-end link
from S2 to S1 via R, denoted by $r_{21}$, satisfies
\begin{eqnarray} \label{eq:rate 21}
r_{21} \leq \frac{1}{2}\log_{2} \left (1 +
\frac{|\mv{h}_1^T\mv{A}\mv{h}_2|^2 p_2 }{ \|\mv{A}^H\mv{h}_1^*\|^2 +
1}\right )
\end{eqnarray}
where the factor $\frac{1}{2}$ is due to the use of two orthogonal
time-slots for relaying. Similarly, it can be shown that the maximum
achievable rate $r_{12}$ for the link from S1 to S2 via R satisfies
\begin{eqnarray} \label{eq:rate 12}
r_{12} \leq \frac{1}{2}\log_{2} \left (1 +
\frac{|\mv{h}_2^T\mv{A}\mv{h}_1|^2 p_1 }{ \|\mv{A}^H\mv{h}_2^*\|^2 +
1}\right ).
\end{eqnarray}

Next, we define the capacity region of ANC-based TWRC,
$\mathcal{C}(P_1,P_2,P_R)$, subject to transmit power constraints at
S1, S2, and R, denoted by $P_1$, $P_2$, and $P_R$, respectively.
First, for a fixed pair of $p_1$ and $p_2$, $p_1\leq P_1$ and
$p_2\leq P_2$, we define the achievable rate region for S1 and S2 as
\begin{eqnarray}\label{eq:rate region}
\mathcal{R}(p_1,p_2,P_R)\triangleq\bigcup_{\mv{A}:~p_R\left(\mv{A}\right)\leq
P_R} \left\{(r_{21},r_{12}): (\ref{eq:rate 21}), (\ref{eq:rate
12})\right\}.
\end{eqnarray}
Then, $\mathcal{C}(P_1,P_2,P_R)$ is defined as
\begin{eqnarray}\label{eq:capacity region}
\mathcal{C}(P_1,P_2,P_R)\triangleq\bigcup_{(p_1,p_2):p_1\leq P_1,
p_2\leq P_2} \mathcal{R}(p_1,p_2,P_R).
\end{eqnarray}
Note that in (\ref{eq:capacity region}), $\mathcal{C}(P_1,P_2,P_R)$
can be obtained by taking the union over all the achievable rate
regions, $\mathcal{R}(p_1,p_2,P_R)$'s, corresponding to different
feasible pairs of $p_1$ and $p_2$. Thus, for the rest of this paper,
we focus our study on characterization of $\mathcal{R}(p_1,p_2,P_R)$
for some fixed $p_1$ and $p_2$. Also note from (\ref{eq:rate
region}) that the relay beamforming matrix $\mv{A}$ plays the role
of realizing different rate tradeoffs between $r_{21}$ and $r_{12}$
on the boundary of $\mathcal{R}(p_1,p_2,P_R)$.

For the convenience of later analysis in this paper, we express
(\ref{eq:Rx signal S1 2}) into an equivalent matrix form as follows,
combined with $y_2(n)$'s and $z_2(n)$'s defined for S2 similarly as
for S1.
\begin{eqnarray}\label{eq:channel matrix form}
\left[\begin{array}{l} y_2(n) \\ y_1(n) \end{array}\right]&=&
\mv{H}_{\rm DL}\mv{A}\mv{H}_{\rm UL} \left[\begin{array}{l}
\sqrt{p_1}s_1(n)
\\ \sqrt{p_2}s_2(n)\end{array}\right]+\mv{H}_{\rm
DL}\mv{A}\mv{z}_R(n) \nonumber \\
&& +\left[\begin{array}{l} z_2(n) \\
z_1(n)
\end{array}\right]
\end{eqnarray}
where $\mv{H}_{\rm UL}=[\mv{h}_1, \mv{h}_2]\in\mathbb{C}^{M\times
2}$ and $\mv{H}_{\rm DL}=[\mv{h}_2,
\mv{h}_1]^T\in\mathbb{C}^{2\times M}$ denote the uplink (UL) and
downlink (DL) channel matrices, respectively. Note that $\mv{H}_{\rm
DL} = \mv{F}\mv{H}_{\rm UL}^T$, where $\mv{F} = \left[{\footnotesize
\begin{array}{cc}
0 & 1\\
1 & 0 \end{array} }\right]$.

\section{Capacity Region Characterization} \label{sec:capacity
region}

In this section, we study the capacity region of ANC-based TWRC by
characterizing $\mathcal{R}(p_1,p_2,P_R)$ defined in (\ref{eq:rate
region}) for a given set of $p_1$, $p_2$, and $P_R$. First, we
derive the optimal relay beamforming structure for $\mv{A}$ that
attains the boundary rate-pairs of $\mathcal{R}(p_1,p_2,P_R)$. It is
shown that with the optimal beamforming structure, the number of
unknown complex-valued variables to be sought in $\mv{A}$ is reduced
from $M^2$ to 4 when $M>2$. Then, we formulate the optimization
problem and present an efficient algorithm to compute the optimal
$\mv{A}$'s to achieve different boundary rate-pairs of
$\mathcal{R}(p_1,p_2,P_R)$.

\subsection{Optimal Relay Beamforming Structure}

Let the singular-value-decomposition (SVD) of $\mv{H}_{\rm UL}$ be
expressed as
\begin{equation}
\mv{H}_{\rm UL}=\mv{U}\mv{\Sigma}\mv{V}^H
\end{equation}
where $\mv{U}\in\mathbb{C}^{M\times 2}$,
$\mv{\Sigma}=\mathtt{diag}(\sigma_1,\sigma_2)$ with $\sigma_1\geq
\sigma_2\geq0$, and $\mv{V}\in\mathbb{C}^{2\times 2}$. It thus
follows that $\mv{H}_{\rm DL}=\mv{F}\mv{V}^*\mv{\Sigma}\mv{U}^T$. We
then have the following theorem.
\begin{theorem}\label{theorem:optimal beamforming}
The optimal relay beamforming matrix, $\mv{A}$, that attains a
boundary rate-pair of $\mathcal{R}(p_1,p_2,P_R)$ defined in
(\ref{eq:rate region}) has the following structure:
\begin{eqnarray}\label{eq:optimal A}
\mv{A} = \mv{U}^* \mv{B} \mv{U}^H
\end{eqnarray}
where $\mv{B}\in\mathbb{C}^{2\times 2}$ is an unknown matrix.
\end{theorem}

\begin{proof}
Please refer to Appendix \ref{appendix:proof beamforming
optimality}.
\end{proof}

\begin{remark}
In the conventional AF-based multi-antenna OWRC, the optimal
beamforming structure at relay to maximize the end-to-end channel
capacity has been studied in, e.g., \cite{Hua07}, \cite{Vidal07}.
Applying the results therein to the OWRC with S2 transmitting to S1
via R yields the optimal $\mv{A}$ to maximize $r_{21}$ in
(\ref{eq:rate 21}) as $\mv{A}_{21}=c_{21}\mv{h}_1^*\mv{h}_2^H$,
where $c_{21}$ is a constant related to $P_R$. Similarly, the
optimal $\mv{A}$ to maximize $r_{12}$ in (\ref{eq:rate 12}) for the
OWRC from S1 to S2 via R is in the form of
$\mv{A}_{12}=c_{12}\mv{h}_2^*\mv{h}_1^H$. It then follows that
$\mv{A}_{21}$ differs from $\mv{A}_{12}$ unless
$\mv{h}_1=\upsilon\mv{h}_2$ for some constant $\upsilon$, i.e.,
$\mv{h}_1$ and  $\mv{h}_2$ are parallel. Therefore, relay
beamforming designs for the OWRC with {\it separate unidirectional}
transmissions in general can not be applied to the TWRC with {\it
simultaneous bidirectional} transmissions. As observed from Theorem
\ref{theorem:optimal beamforming}, the optimal relay beamforming
matrix for TWRC lies in the space spanned by both $\mv{h}_1$ and
$\mv{h}_2$.
\end{remark}

Let $\mv{g}_1=\mv{U}^H\mv{h}_1\in\mathbb{C}^{2\times1}$ and
$\mv{g}_2=\mv{U}^H\mv{h}_2\in\mathbb{C}^{2\times1}$ be the
``effective'' channels from S1 to R and from S2 to R, respectively,
by applying the optimal structure of $\mv{A}$ given in
(\ref{eq:optimal A}). Similarly, $\mv{g}_1^T$ and $\mv{g}_2^T$
become the effective channels from R to S1 and S2, respectively.
$\mathcal{R}(p_1,p_2,P_R)$ in (\ref{eq:rate region}) can then be
equivalently re-expressed as
{\small \begin{align}\label{eq:rate
region 2} \bigcup_{\mv{B}: ~p_R\left(\mv{B}\right)\leq P_R} \bigg\{
& (r_{21},r_{12}): r_{21} \leq \frac{1}{2}\log_{2} \left (1 +
\frac{|\mv{g}_1^T\mv{B}\mv{g}_2|^2 p_2 }{ \|\mv{B}^H\mv{g}_1^*\|^2 +
1}\right ), \nonumber \\ & r_{12} \leq \frac{1}{2}\log_{2} \left (1
+ \frac{|\mv{g}_2^T\mv{B}\mv{g}_1|^2 p_1 }{ \|\mv{B}^H\mv{g}_2^*\|^2
+ 1}\right ) \bigg\}
\end{align}}where  $p_R(\mv{B})=\|\mv{B}\mv{g}_1 \|^{2} p_{1} + \|
\mv{B}\mv{g}_2 \|^{2} p_2 + \mathtt{tr} (\mv{B}\mv{B}^{H})$. Note
that the not-yet-determined parameter in (\ref{eq:rate region 2}) is
$\mv{B}$. Since $\mv{B}$ has 4 complex-valued variables as compared
to $M^2$ in $\mv{A}$, the complexity for searching the optimal
$\mv{B}$ corresponding to a particular boundary rate-pair of
$\mathcal{R}(p_1,p_2,P_R)$ is reduced when $M>2$. Using Theorem
\ref{theorem:optimal beamforming} and (\ref{eq:rate region 2}),
optimal structures of $\mv{A}$ can be further simplified in the
following two special cases, which are {\it Case I}:
$\mv{h}_1\bot\mv{h}_2$, i.e., $\mv{h}_1^H\mv{h}_2=0$; and {\it Case
II}: $\mv{h}_1\parallel \mv{h}_2$, i.e., $\mv{h}_1=\upsilon\mv{h}_2$
with $\upsilon$ being a constant.
\begin{lemma}\label{lemma:zero correlation}
In the case of $\mv{h}_1\bot\mv{h}_2$, the optimal structure of
$\mv{A}$ is in the form of $\mv{A}=\mv{U}^*\left[{\footnotesize
\begin{array}{cc}
0 & c\\
d & 0 \end{array} }\right]\mv{U}^H$, with $c\geq 0$ and $d\geq 0$.
\end{lemma}

\begin{proof}
Please refer to Appendix \ref{appendix:proof zero correlation}.
\end{proof}

\begin{lemma}\label{lemma:unit correlation}
In the case of $\mv{h}_1\parallel\mv{h}_2$, the optimal structure of
$\mv{A}$ is in the form of $\mv{A}=\mv{U}^*\left[{\footnotesize
\begin{array}{cc}
a & 0\\
0 & 0 \end{array} }\right]\mv{U}^H$, with $a\geq 0$.
\end{lemma}

\begin{proof}
Please refer to Appendix \ref{appendix:proof unit correlation}.
\end{proof}

Note that in other cases of $\mv{h}_1$ and $\mv{h}_2$ beyond the
above two, we do not have further simplified structures for
$\mv{B}$, or $\mv{A}$ upon that in (\ref{eq:optimal A}). Thus, in
general we need to resort to optimization techniques to obtain the
$2\times 2$ matrix $\mv{B}$ for each boundary rate-pair of
$\mathcal{R}(p_1,p_2,P_R)$, as will be shown next.

\subsection{Optimization Problems}

Since $\mathcal{R}(p_1,p_2,P_R)$ in (\ref{eq:rate region 2}) is the
same as that in (\ref{eq:rate region}), we use (\ref{eq:rate region
2}) in this subsection to characterize all the boundary rate-pairs
of $\mathcal{R}(p_1,p_2,P_R)$. A commonly used method to
characterize different rate-tuples on the boundary of a multiuser
capacity region is via solving a sequence of {\it weighted sum-rate
maximization} (WSRMax) problems, each for a different (nonnegative)
rate weight vector of users. In the case of TWRC, let
$\mv{w}=[w_{21}, w_{12}]^T$ be the
 weight vector, where $w_{21}$ and $w_{12}$ are
the ``rate rewards'' for $r_{21}$ and $r_{12}$, respectively. From
(\ref{eq:rate region 2}), we can express the WSRMax problem to
determine a particular boundary rate-pair of
$\mathcal{R}(p_1,p_2,P_R)$ as
\begin{align}\label{eq:WSRMax}
\mathop{\mathtt{Max.}}_{\mv{B}} & ~~\frac{w_{21}}{2}\log_{2} \left
(1 + \frac{|\mv{g}_1^T\mv{B}\mv{g}_2|^2 p_2 }{
\|\mv{B}^H\mv{g}_1^*\|^2 + 1}\right) \nonumber  \\ &~~~~
+\frac{w_{12}}{2}\log_{2} \left (1 +
\frac{|\mv{g}_2^T\mv{B}\mv{g}_1|^2 p_1 }{ \|\mv{B}^H\mv{g}_2^*\|^2 +
1}\right ) \nonumber \\
\mathtt{s.t.} &~~\|\mv{B}\mv{g}_1 \|^{2} p_{1} + \| \mv{B}\mv{g}_2
\|^{2} p_2 + \mathtt{tr} (\mv{B}\mv{B}^{H})\leq P_R.
\end{align}
In the above problem, although the constraint is convex, the
objective function is not a concave function of $\mv{B}$. As a
result, this problem is non-convex \cite{Boydbook}, and is thus
difficult to solve via standard convex optimization techniques.

Therefore, we need to resort to an alternative method of WSRMax to
characterize $\mathcal{R}(p_1,p_2,P_R)$. In \cite{Mohseni}, an
interesting concept so-called {\it rate profile} was introduced to
efficiently characterize boundary rate-tuples of a capacity region.
A rate profile regulates the ratio between each user's rate, $r_k$,
and their sum-rate, $R_{\rm sum}=\sum_{k=1}^K r_k$, to be a
predefined value $\alpha_k$, i.e., $\frac{r_k}{R_{\rm
sum}}=\alpha_k, k=1,\ldots,K$, with $K$ denoting the number of
users. The rate-profile vector is then defined as
$\mv{\alpha}=[\alpha_1,\ldots,\alpha_K]^T$. For a given
$\mv{\alpha}$, if $R_{\rm sum}$ is maximized subject to the
rate-profile constraint specified by $\mv{\alpha}$, the solution
rate-tuple, $R_{\rm sum}\mv{\alpha}$, can then be geometrically
viewed as the intersection of a straight line specified by a slope
of $\mv{\alpha}$ and passing through the origin of the capacity
region, with the capacity region boundary. Thereby, with different
$\mv{\alpha}$'s, all the boundary rate-tuples of the capacity region
can be obtained.

Next, we show that by applying the above method based on rate
profile, boundary rate-pairs of $\mathcal{R}(p_1,p_2,P_R)$ can be
efficiently characterized. Since in our case
$\mathcal{R}(p_1,p_2,P_R)$ lies in a two-dimensional space, we can
express the rate-profile vector as $\mv{\alpha}=[\alpha_{21},
\alpha_{12}]^T$, where $\alpha_{21}=\frac{r_{21}}{R_{\rm sum}}$,
$\alpha_{12}=\frac{r_{12}}{R_{\rm sum}}$, and $R_{\rm
sum}=r_{21}+r_{12}$.  For a fixed $\mv{\alpha}$, we consider the
following sum-rate maximization problem:
\begin{align}\label{eq:SRMax}
\mathop{\mathtt{Max.}}_{R_{\rm sum}, \mv{B}}& ~~ R_{\rm sum} \nonumber \\
\mathtt{s.t.} & ~~ \frac{1}{2}\log_{2} \left (1 +
\frac{|\mv{g}_1^T\mv{B}\mv{g}_2|^2 p_2 }{ \|\mv{B}^H\mv{g}_1^*\|^2 +
1}\right)\geq \alpha_{21}R_{\rm sum} \nonumber \\
& ~~ \frac{1}{2}\log_{2} \left (1 +
\frac{|\mv{g}_2^T\mv{B}\mv{g}_1|^2 p_1 }{ \|\mv{B}^H\mv{g}_2^*\|^2 +
1}\right )\geq \alpha_{12}R_{\rm
sum} \nonumber \\
& ~~ \|\mv{B}\mv{g}_1 \|^{2} p_{1} + \| \mv{B}\mv{g}_2 \|^{2} p_2 +
\mathtt{tr} (\mv{B}\mv{B}^{H})\leq P_R.
\end{align}
After solving the above problem, solution of $\mv{B}$ can be used to
construct the optimal relay beamforming matrix $\mv{A}$ according to
(\ref{eq:optimal A}), and the solution $R_{\rm sum}\mv{\alpha}$
becomes the rate-pair on the boundary of $\mathcal{R}(p_1,p_2,P_R)$
corresponding to the given $\mv{\alpha}$. To solve problem
(\ref{eq:SRMax}), we first consider the following relay power
minimization problem subject to rate constraints:
\begin{align}\label{eq:PowerMin}
\mathop{\mathtt{Min.}}_{\mv{B}}& ~~ p_R:=\|\mv{B}\mv{g}_1 \|^{2}
p_{1} +\| \mv{B}\mv{g}_2 \|^{2} p_2 + \mathtt{tr} (\mv{B}\mv{B}^{H})\nonumber \\
\mathtt{s.t.} & ~~ \frac{1}{2}\log_{2} \left (1 +
\frac{|\mv{g}_1^T\mv{B}\mv{g}_2|^2 p_2 }{ \|\mv{B}^H\mv{g}_1^*\|^2 +
1}\right)\geq \alpha_{21}r \nonumber \\
&~~ \frac{1}{2}\log_{2} \left (1 +
\frac{|\mv{g}_2^T\mv{B}\mv{g}_1|^2 p_1 }{ \|\mv{B}^H\mv{g}_2^*\|^2 +
1}\right )\geq \alpha_{12}r.
\end{align}
If the above problem is feasible, its optimal value, denoted by
$p_R^{\star}$, will be the minimum relay power required to support
the given rate-pair $r\mv{\alpha}$; otherwise, there is no finite
relay power that can support this rate-pair, and for convenience we
denote $p_R^{\star}=+\infty$ in this case. Problems (\ref{eq:SRMax})
and (\ref{eq:PowerMin}) are related as follows. If for some given
$r$, $\mv{\alpha}$ and $P_R$, the optimal value of problem
(\ref{eq:PowerMin}) satisfies that $p_R^{\star}>P_R$, it follows
that $r$ must be an infeasible solution of $R_{\rm sum}$ in problem
(\ref{eq:SRMax}), i.e., the rate-pair $r\mv{\alpha}$ must fall
outside $\mathcal{R}(p_1,p_2,P_R)$ along the line specified by slope
$\mv{\alpha}$; if $p_R^{\star}\leq P_R$, it follows that $r$ is a
feasible solution of $R_{\rm sum}$ and thus $r\mv{\alpha}$ must be
within $\mathcal{R}(p_1,p_2,P_R)$. Based on the above observations,
we obtain the following algorithm for problem (\ref{eq:SRMax}), for
which a rigorous proof is given in Appendix \ref{appendix:proof
algorithm}.
\begin{algorithm}\label{algorithm:SR max}
$ $
\begin{itemize}
\item {\bf Given} $R_{\rm sum} \in [0, \bar{R}_{\rm
sum}]$, $\mv{\alpha}$.
\item {\bf Initialize} $r_{\min}=0$, $r_{\max}=\bar{R}_{\rm sum}$.
\item {\bf Repeat}
\begin{itemize}
\item[1.] Set $r\leftarrow\frac{1}{2}(r_{\min}+r_{\max})$.
\item[2.] Solve problem (\ref{eq:PowerMin}) to obtain its optimal value, $p_R^{\star}$.
\item[3.] Update $r$ by the bisection method \cite{Boydbook}: If
$p_R^{\star}\leq P_R$, set $r_{\min}\leftarrow r$; otherwise,
$r_{\max}\leftarrow r$.
\end{itemize}
\item {\bf Until} $r_{\max}-r_{\min}\leq\delta_{r}$, where $\delta_{r}$ is a small positive constant
to control the algorithm accuracy. The converged value of $r_{\min}$
is the optimal solution of $R_{\rm sum}$ in (\ref{eq:SRMax}).
\end{itemize}
\end{algorithm}

Note that $\bar{R}_{\rm sum}$ is an upper bound on the optimal
solution of $R_{\rm sum}$ in (\ref{eq:SRMax}) for the given
$\mv{\alpha}$. In Section \ref{sec:suboptimal schemes} (see Remark
\ref{remark:sum rate upper bound}), we obtain such an upper bound
that is valid for all possible values of $\mv{\alpha}$. In the next
subsection, we will address the remaining part in Algorithm
\ref{algorithm:SR max} on how to solve problem (\ref{eq:PowerMin})
in Step 2.

\subsection{Power Minimization under SNR Constraints}

Denote $\gamma_{1}$ and $\gamma_2$ as the SNRs at the receivers of
S1 and S2, respectively, which are defined as
\begin{eqnarray}\label{eq:SNR}
\gamma_1=\frac{|\mv{g}_1^T\mv{B}\mv{g}_2|^2 p_2 }{
\|\mv{B}^H\mv{g}_1^*\|^2+1}, \
\gamma_2=\frac{|\mv{g}_2^T\mv{B}\mv{g}_1|^2 p_1 }{
\|\mv{B}^H\mv{g}_2^*\|^2 + 1}.
\end{eqnarray}
Let $\bar{\gamma}_1=2^{2\alpha_{21}r}-1$ and
$\bar{\gamma}_2=2^{2\alpha_{12}r}-1$ be the equivalent SNR targets
at S1 and S2 to guarantee the given rate constraints. Then, it is
observed that the rate constraints in (\ref{eq:PowerMin}) can be
expressed as the corresponding SNR constraints at S1 and S2,
$\gamma_1\geq \bar{\gamma}_1$ and $\gamma_2\geq \bar{\gamma}_2$,
respectively. Using (\ref{eq:SNR}), problem (\ref{eq:PowerMin}) can
be recast as the following equivalent problem:
\begin{align}\label{eq:PowerMin SNR Constraints}
\mathop{\mathtt{Min.}}_{\mv{B}} & ~~ p_R:=\|\mv{B}\mv{g}_1 \|^{2}
p_{1} + \| \mv{B}\mv{g}_2 \|^{2} p_2 + \mathtt{tr} (\mv{B}\mv{B}^{H})\nonumber \\
\mathtt{s.t.} & ~~ |\mv{g}_1^T\mv{B}\mv{g}_2|^2 \geq
\frac{\bar{\gamma}_1}{p_2}\|\mv{B}^H\mv{g}_1^*\|^2+\frac{\bar{\gamma}_1}{p_2} \nonumber \\
& ~~ |\mv{g}_2^T\mv{B}\mv{g}_1|^2 \geq
\frac{\bar{\gamma}_2}{p_1}\|\mv{B}^H\mv{g}_2^*\|^2 +
\frac{\bar{\gamma}_2}{p_1}.
\end{align}
Note that the above problem may be of practical interest itself,
since it is relevant when certain prescribed transmission
quality-of-service (QoS) requirements in terms of receiver SNRs need
to be fulfilled at S1 and S2. For the convenience of analysis, we
modify the above problem as follows. First, let
$\mathtt{Vec}(\mv{Q})$ be a $K^2\times 1$ vector associated with a
$K\times K$ square matrix $\mv{Q}=[\mv{q}_1,\ldots,\mv{q}_K]^T$,
where $\mv{q}_k\in\mathbb{C}^{K \times 1}, k=1,\ldots,K$, by the
rule of $\mathtt{Vec}(\mv{Q})=[\mv{q}_1^T,\ldots, \mv{q}_K^T]^T$.
Next, with $\mv{b}=\mathtt{Vec}(\mv{B})$ and
$\mv{\Theta}=p_1\mv{g}_1\mv{g}_1^H+p_2\mv{g}_2\mv{g}_2^H+\mv{I}$, we
can express $p_R$ in the objective function of (\ref{eq:PowerMin SNR
Constraints}) as $p_R=\mathtt{tr}(\mv{B\Theta
B}^H)=\|\mv{\Phi}\mv{b}\|^2$, where
$\mv{\Phi}=(\mathtt{diag}(\mv{\Theta}^T,\mv{\Theta}^T))^{\frac{1}{2}}$.
Similarly, let
$\mv{f}_1=\mathtt{Vec}\left(\mv{g}_1\mv{g}_2^T\right)$ and
$\mv{f}_2=\mathtt{Vec}\left(\mv{g}_2\mv{g}_1^T\right)$. Then, from
(\ref{eq:PowerMin SNR Constraints}) it follows that
$|\mv{g}_1^T\mv{B}\mv{g}_2|^2=|\mv{f}_1^T\mv{b}|^2$ and
$|\mv{g}_2^T\mv{B}\mv{g}_1|^2=|\mv{f}_2^T\mv{b}|^2$. Furthermore, by
defining
\begin{eqnarray*}
\mv{G}_i=\left[\begin{array}{cccc}
\mv{g}_i(1,1) & 0 & \mv{g}_i(2,1) & 0 \\
0 &  \mv{g}_i(1,1)  & 0 & \mv{g}_i(2,1) \end{array} \right], \
i=1,2,
\end{eqnarray*}
we have $\|\mv{B}^H\mv{g}_i^*\|^2=\|\mv{G}_i\mv{b}\|^2, i=1,2$.
Using the above transformations, (\ref{eq:PowerMin SNR Constraints})
can be rewritten as
\begin{align}\label{eq:PowerMin SNR Constraints 2}
\mathop{\mathtt{Min.}}_{\mv{b}} & ~~ p_R:=\|\mv{\Phi}\mv{b}\|^2\nonumber \\
\mathtt{s.t.} & ~~ |\mv{f}_1^T\mv{b}|^2 \geq
\frac{\bar{\gamma}_1}{p_2}\|\mv{G}_1\mv{b}\|^2+\frac{\bar{\gamma}_1}{p_2} \nonumber \\
&~~ |\mv{f}_2^T\mv{b}|^2 \geq
\frac{\bar{\gamma}_2}{p_1}\|\mv{G}_2\mv{b}\|^2 +
\frac{\bar{\gamma}_2}{p_1}.
\end{align}
The above problem can be shown to be still non-convex. However, in
the following, we show that the exact optimal solution could be
obtained via a relaxed semidefinite programming
(SDP)~\cite{Boydbook} problem.

We first define $\mv{E}_0=\mv{\Phi}^H\mv{\Phi}$,
$\mv{E}_1=\frac{p_2}{\bar{\gamma}_1}\mv{f}_1^*\mv{f}_1^T-\mv{G}_1^H\mv{G}_1$,
and
$\mv{E}_2=\frac{p_1}{\bar{\gamma}_2}\mv{f}_2^*\mv{f}_2^T-\mv{G}_2^H\mv{G}_2$.
Since standard SDP formulations only involve real variables and
constants, we introduce a new real matrix variable as
$\mv{X}=[\mv{b}_R; \mv{b}_I]\times[\mv{b}_R; \mv{b}_I]^T$, where
$\mv{b}_R=Re(\mv{b})$ and $\mv{b}_I=Im(\mv{b})$ are the real and
imaginary parts of $\mv{b}$, respectively. To rewrite the norm
representations at~(\ref{eq:PowerMin SNR Constraints 2}) in terms of
$\mv{X}$, we need to rewrite $\mv{E}_0$, $\mv{E}_1$, and $\mv{E}_2$,
as expanded matrices $\mv{F}_0$, $\mv{F}_1$, and $\mv{F}_2$,
respectively, in terms of their real and imaginary parts.
Specifically, to write out $\mv{F}_0$, we first define the short
notations $\mv{\Phi}_R=Re(\mv{\Phi})$ and
$\mv{\Phi}_I=Im(\mv{\Phi})$; then we have
\begin{eqnarray}
\mv{F}_0&=& \left[\begin{array}{cc}
\mv{\Phi}_R^T\mv{\Phi}_R+\mv{\Phi}_I^T\mv{\Phi}_I & \mv{\Phi}_I^T\mv{\Phi}_R-\mv{\Phi}_R^T\mv{\Phi}_I\\
\mv{\Phi}_R^T\mv{\Phi}_I-\mv{\Phi}_I^T\mv{\Phi}_R & \mv{\Phi}_R^T\mv{\Phi}_R+\mv{\Phi}_I^T\mv{\Phi}_I \end{array}\right]. \nonumber
\end{eqnarray}
The expanded matrices $\mv{F}_1$ and $\mv{F}_2$ can be generated
from $\mv{E}_1$ and $\mv{E}_2$ in a similar way, where the two terms
in $\mv{E}_1$ or $\mv{E}_2$ could first be expanded separately then
summed together.

As such, problem (\ref{eq:PowerMin SNR Constraints 2}) can be
equivalently rewritten as
\begin{align}\label{eq:PowerMin SNR SDR}
\mathop{\mathtt{Min.}}_{\mv{X}} &~~ p_R:=\mathtt{tr}(\mv{F}_0\mv{X}) \nonumber \\
\mathtt{s.t.} &~~ \mathtt{tr}(\mv{F}_1\mv{X})\geq 1, ~
\mathtt{tr}(\mv{F}_2\mv{X})\geq 1, ~ \mv{X}\succeq 0, \nonumber \\
&~~ \mathtt{rank}(\mv{X})=1.
\end{align}
The above problem is still not convex given the last rank-one
constraint. However, if we remove such a constraint, this problem is
relaxed into a convex SDP problem as shown below.
\begin{align}\label{eq:PowerMin_SDP}
\mathop{\mathtt{Min.}}_{\mv{X}} &~~ p_R:=\mathtt{tr}(\mv{F}_0\mv{X}) \nonumber \\
\mathtt{s.t.} &~~ \mathtt{tr}(\mv{F}_1\mv{X})\geq 1, ~
\mathtt{tr}(\mv{F}_2\mv{X})\geq 1, ~ \mv{X}\succeq 0.
\end{align}
Given the convexity of the above SDP problem, the optimal solution
could be efficiently found by various convex optimization
methods~\cite{Boydbook}. Note that if problem
(\ref{eq:PowerMin_SDP}) is infeasible, so is the more restricted
problem (\ref{eq:PowerMin SNR SDR}). Thus, we assume w.l.o.g. that
problem (\ref{eq:PowerMin_SDP}) is feasible in the following
discussions. SDP relaxation usually leads to an optimal $\mv{X}$ for
problem (\ref{eq:PowerMin_SDP}) that is of rank $r$ with $r\ge{1}$,
which makes it impossible to reconstruct the exact optimal solution
for problem~(\ref{eq:PowerMin SNR Constraints 2}) when $r>1$. A
commonly adopted method in the literature to obtain a feasible
rank-one (but in general suboptimal) solution from the solution of
SDP relaxation is via ``randomization'' (see, e.g., \cite{Luo06} and
references therein). Fortunately, we show in the following that with
the special structure in problem~(\ref{eq:PowerMin_SDP}), we could
efficiently reconstruct an optimal rank-one solution from its
optimal solution that could be of rank $r$ with $r>1$, based on some
elegant results derived for SDP relaxation in~\cite{Ye_Zhang}. In
other words, we could obtain the exact optimal solution for the
non-convex problem in~(\ref{eq:PowerMin SNR SDR}) without losing any
optimality, and as efficiently as solving a convex problem.

\begin{theorem}\label{theorem:rank_one}
Assume that an optimal solution $\mv{X}^\star$ of rank $r>1$ has
been found for problem~(\ref{eq:PowerMin_SDP}), we could efficiently
construct another feasible optimal solution $\mv{X}^{\star\star}$ of
rank one, i.e., $\mv{X}^{\star\star}$ is the optimal solution for
both (\ref{eq:PowerMin SNR SDR}) and (\ref{eq:PowerMin_SDP}).
\end{theorem}

\begin{proof} Please refer to Appendix \ref{appendix:proof rank
one}.
\end{proof}

Since the above proof is self-constructive, we could write a routine
to obtain an optimal rank-one solution for problem (\ref{eq:PowerMin
SNR SDR}) from $\mv{X}^\star$, as given in the last part of Appendix
\ref{appendix:proof rank one}.

\section{Low-Complexity Relay Beamforming Schemes} \label{sec:suboptimal
schemes}

In this section, we present suboptimal relay beamforming schemes
that require lower complexity for implementation than the optimal
scheme developed in Section \ref{sec:capacity region}. Two
suboptimal beamforming structures for $\mv{A}$ are proposed as
follows:
\begin{itemize}
\item Maximal-Ratio Reception and Maximal-Ratio Transmission
(MRR-MRT):
\begin{eqnarray}\label{eq:MRR MRT}
\mv{A}_{\rm MR}=\mv{H}_{\rm DL}^H \left[\begin{array}{ll}
a_{\rm MR} & 0\\
0 & b_{\rm MR} \end{array}\right]\mv{H}_{\rm UL}^H;
\end{eqnarray}

\item Zero-Forcing Reception and Zero-Forcing Transmission
(ZFR-ZFT):\footnote{Note that the ZFR-ZFT scheme with $a_{\rm
ZF}=b_{\rm ZF}$ has also been proposed in \cite{Unger2007}, but
without detailed performance analysis.}
\begin{eqnarray}\label{eq:ZFR ZFT}
\mv{A}_{\rm ZF}=\mv{H}_{\rm DL}^{\dag} \left[\begin{array}{cc}
a_{\rm ZF} & 0\\
0 & b_{\rm ZF} \end{array}\right] \mv{H}_{\rm UL}^{\dag}.
\end{eqnarray}
\end{itemize}

Note that from (\ref{eq:channel matrix form}), it follows that
$a_{\rm x}$ and $b_{\rm x}$, ${\rm x=MR}$ or ${\rm ZF}$, $a_{\rm
x}\geq 0$ and $b_{\rm x}\geq 0$, in the above beamforming structures
play the role of balancing relay power allocations to transmissions
from S1 to S2 and from S2 to S1. MRR-MRT applies the
``matched-filter (MF)'' -based receive and transmit beamforming at R
to maximize the total signal power forwarded to S1 and S2. However,
in this scheme, R does not attempt to suppress or mitigate the
interference between S1 and S2. On the other hand, ZFR-ZFT applies
the ``zero-forcing (ZF)'' -based receive and transmit beamforming to
remove the interferences between S1 and S2 at R as well as at the
end receivers of S1 and S2. To illustrate this, we substitute
$\mv{A}_{\rm ZF}$ in (\ref{eq:ZFR ZFT}) into (\ref{eq:channel matrix
form}) to obtain $\left[\begin{array}{l} y_2(n) \\ y_1(n)
\end{array}\right]$ in the form of
\begin{eqnarray*}\label{eq:channel matrix form ZF}
\left[\begin{array}{l} a_{\rm ZF}\sqrt{p_1}s_1(n)
\\ b_{\rm ZF}\sqrt{p_2}s_2(n)\end{array}\right]+\left[\begin{array}{cc}
a_{\rm ZF} & 0\\
0 & b_{\rm ZF} \end{array}\right] \mv{H}_{\rm UL}^{\dag}
\mv{z}_R(n)+\left[\begin{array}{l} z_2(n) \\ z_1(n)
\end{array}\right].
\end{eqnarray*}
It is observed from the above that the self-interferences are
completely removed at the receivers of S1 and S2 by ZF-based relay
beamforming. Therefore, the main advantage of ZFR-ZFT over MRR-MRT
lies in that it does not need to implement the self-interference
cancelation at S1 or S2, and thus simplifies their receivers. In
general, with ANC, we know that the interference between S1 and S2
observed at R is in fact the self-interference of S1 or S2, and can
be later removed at the end receiver of S1 or S2. Thus, it is
conjectured that MRR-MRT may have a superior performance over
ZFR-ZFT for ANC-based TWRC. This conjecture is in fact true, and
will be verified in later parts of this paper via performance
analysis and simulation results.

Interestingly, the above two suboptimal beamforming schemes both
comply with the optimal beamforming structure given in
(\ref{eq:optimal A}), while their associated values of $\mv{B}$ are
in general suboptimal. This can be easily verified by rewriting
$\mv{A}_{\rm ZF}$ in (\ref{eq:MRR MRT}) and $\mv{A}_{\rm MR}$ in
(\ref{eq:ZFR ZFT}) as $\mv{A}_{\rm MR}=\mv{U}^*\mv{B}_{\rm
MR}\mv{U}^H$ and $\mv{A}_{\rm ZF}=\mv{U}^*\mv{B}_{\rm ZF}\mv{U}^H$,
respectively, where
\begin{eqnarray}
\mv{B}_{\rm MR} &=& \mv{\Sigma}\mv{V}^T\left[\begin{array}{cc}
0 & a_{\rm MR} \\
b_{\rm MR} & 0 \end{array}\right]\mv{V}\mv{\Sigma} \label{eq:B
MF} \\
\mv{B}_{\rm ZF} &=& \mv{\Sigma}^{-1}\mv{V}^T\left[\begin{array}{cc}
0 & a_{\rm ZF} \\
b_{\rm ZF} & 0 \end{array}\right]\mv{V}\mv{\Sigma}^{-1}. \label{eq:B
ZF}
\end{eqnarray}

Using Lemmas \ref{lemma:zero correlation} and \ref{lemma:unit
correlation}, we can show the following results on the optimality of
MRR-MRT and ZFR-ZFT in some special channel cases. For brevity, here
we omit the proofs.
\begin{lemma}\label{lemma:optimality MR}
In both cases of $\mv{h}_1\bot\mv{h}_2$ and
$\mv{h}_1\parallel\mv{h}_2$, $\mv{A}_{\rm MR}$ in (\ref{eq:MRR MRT})
is equivalent to the optimal $\mv{A}$ given in (\ref{eq:optimal A}).
\end{lemma}

\begin{lemma}\label{lemma:optimality ZF}
In the case of $\mv{h}_1\bot\mv{h}_2$, $\mv{A}_{\rm ZF}$ in
(\ref{eq:ZFR ZFT}) is equivalent to the optimal $\mv{A}$ given in
(\ref{eq:optimal A}).
\end{lemma}

It is also noted that in the case of $\mv{h}_1\parallel \mv{h}_2$,
$\mv{B}_{\rm ZF}$ in (\ref{eq:B ZF}) does not exist since in
$\mv{\Sigma}$, $\sigma_2=0$, and thus $\mv{\Sigma}$ is
non-invertible. As a result, $\mv{A}_{\rm ZF}$ does not exist either
in this case.

\section{Performance Analysis} \label{sec:performances}

To further investigate the performances of the proposed optimal and
suboptimal relay beamforming schemes, we study in this section their
achievable sum-rates in TWRC. First, we derive an upper bound on the
maximum sum-rate or the sum-capacity achievable by the optimal
beamforming scheme, as well as various lower bounds on the
achievable sum-rates by the suboptimal schemes. Then, by comparing
these rate bounds at asymptotically high SNR, we characterize the
limiting sum-rate losses resulted by the suboptimal beamforming
schemes as compared to the sum-capacity.

\subsection{Rate Bounds}

First, we study the sum-capacity of TWRC with given $P_R$, $p_1$,
and $p_2$. The sum-capacity of TWRC can be obtained by solving the
WSRMax problem (\ref{eq:WSRMax}) with $w_{12}=w_{21}=1$. Since
WSRMax for TWRC is non-convex and is thus difficult to solve, we
consider an upper bound on the sum-capacity, which can be obtained
by solving the following modified problem of (\ref{eq:WSRMax}):
\begin{align}\label{eq:SRMax UB}
\mathop{\mathtt{Max.}}_{\mv{B}_{21}, \mv{B}_{12}} &~~
\frac{1}{2}\log_{2} \left (1 +
\frac{|\mv{g}_1^T\mv{B}_{21}\mv{g}_2|^2 p_2 }{
\|\mv{B}_{21}^H\mv{g}_1^*\|^2 + 1}\right) \nonumber
\\ &~~ +\frac{1}{2}\log_{2} \left (1 +
\frac{|\mv{g}_2^T\mv{B}_{12}\mv{g}_1|^2 p_1 }{
\|\mv{B}_{12}^H\mv{g}_2^*\|^2 +
1}\right ) \nonumber \\
\mathtt{s.t.} &~~ \|\mv{B}_{12}\mv{g}_1 \|^{2} p_{1} + \|
\mv{B}_{21}\mv{g}_2 \|^{2} p_2 + \kappa_{12}\mathtt{tr}
(\mv{B}_{12}\mv{B}_{12}^{H}) \nonumber \\
&~~ +\kappa_{21}\mathtt{tr} (\mv{B}_{21}\mv{B}_{21}^{H}) \leq P_R
\end{align}
where $\kappa_{12}$ and $\kappa_{21}$ are nonnegative and satisfy
$\kappa_{12}+\kappa_{21}=1$. Let $C_{\rm
sum}(\kappa_{12},\kappa_{21})$ denote the maximum value of the above
problem. Note that if we add the constraint
$\mv{B}_{12}=\mv{B}_{21}=\mv{B}$ into the above problem, solution of
$\mv{B}$ will lead to the exact sum-capacity of TWRC. Since $C_{\rm
sum}(\kappa_{12},\kappa_{21})$ is an upper bound on the sum-capacity
for any feasible $\kappa_{12}$ and $\kappa_{21}$, it can be
tightened by minimizing $C_{\rm sum}(\kappa_{12},\kappa_{21})$ over
all the feasible pairs of $\kappa_{12}$ and $\kappa_{21}$. For given
$\kappa_{12}$ and $\kappa_{21}$, problem (\ref{eq:SRMax UB}) can be
decomposed into the following two independent subproblems:
\begin{align}\label{eq:SRMax UB subproblem 1}
\mathop{\mathtt{Max.}}_{\mv{B}_{21}} &~~ \frac{1}{2}\log_{2} \left
(1 + \frac{|\mv{g}_1^T\mv{B}_{21}\mv{g}_2|^2 p_2 }{
\|\mv{B}_{21}^H\mv{g}_1^*\|^2 + 1}\right) \nonumber \\
\mathtt{s.t.} &~~  \| \mv{B}_{21}\mv{g}_2 \|^{2} p_2 +
\kappa_{21}\mathtt{tr} (\mv{B}_{21}\mv{B}_{21}^{H}) \leq P_{21}
\end{align}
\begin{align}\label{eq:SRMax UB subproblem 2}
\mathop{\mathtt{Max.}}_{\mv{B}_{12}} &~~ \frac{1}{2}\log_{2} \left
(1 + \frac{|\mv{g}_2^T\mv{B}_{12}\mv{g}_1|^2 p_1 }{
\|\mv{B}_{12}^H\mv{g}_2^*\|^2 +
1}\right ) \nonumber \\
\mathtt{s.t.} &~~ \|\mv{B}_{12}\mv{g}_1 \|^{2} p_{1} +
\kappa_{12}\mathtt{tr} (\mv{B}_{12}\mv{B}_{12}^{H})  \leq P_{12}
\end{align}
subject to an additional common constraint $P_{21}+P_{12}\leq P_R$.
Let $C_{21}(\kappa_{21},P_{21})$ and $C_{12}(\kappa_{12},P_{12})$
denote the maximum values of the above two subproblems,
respectively. Thus, $C_{\rm sum}(\kappa_{12},\kappa_{21})$ can be
obtained by first solving the above two subproblems for given
$P_{12}$ and $P_{21}$, and then maximizing
$C_{21}(\kappa_{21},P_{21})+C_{12}(\kappa_{12},P_{12})$ over all the
feasible values of $P_{12}$ and $P_{21}$. Note that each of the
above two subproblems optimizes the relay beamforming matrix to
maximize the capacity of the corresponding OWRC from S2 and S1 via
R, or from S1 to S2 via R. By applying the results in prior work
\cite{Hua07}, \cite{Vidal07}, the optimal solutions to
(\ref{eq:SRMax UB subproblem 1}) and (\ref{eq:SRMax UB subproblem
2}) can be obtained as
\begin{eqnarray}
\mv{B}_{21}=\sqrt{\frac{P_{21}}{p_2\|\mv{g}_2\|^2+\kappa_{21}}}\tilde{\mv{g}}_1^*\tilde{\mv{g}}_2^H
\label{eq:optimal B OWRC 1}
\\
\mv{B}_{12}=\sqrt{\frac{P_{12}}{p_1\|\mv{g}_1\|^2+\kappa_{12}}}\tilde{\mv{g}}_2^*\tilde{\mv{g}}_1^H
\label{eq:optimal B OWRC 2}
\end{eqnarray}
where $\tilde{\mv{g}}_i=\frac{\mv{g}_i}{\|\mv{g}_i\|}, i=1,2$. By
substituting the above expressions into the objective functions of
(\ref{eq:SRMax UB subproblem 1}) and (\ref{eq:SRMax UB subproblem
2}), respectively, we obtain
\begin{eqnarray}
C_{21}(\kappa_{21},P_{21})=\frac{1}{2}\log_{2} \left (1 +
\frac{\theta_2p_2}{1+\frac{(\theta_2/\theta_1)p_2}{P_{21}}+\frac{\kappa_{21}}{\theta_1P_{21}}}\right)
\label{eq:C21}
\\
C_{12}(\kappa_{12},P_{12})=\frac{1}{2}\log_{2} \left (1 +
\frac{\theta_1p_1}{1+\frac{(\theta_1/\theta_2)p_1}{P_{12}}+\frac{\kappa_{12}}{\theta_2P_{12}}}\right)
\label{eq:C12}
\end{eqnarray}
where, for conciseness, we have denoted
$\|\mv{g}_1\|^2=\|\mv{h}_1\|^2=\theta_1$ and
$\|\mv{g}_2\|^2=\|\mv{h}_2\|^2=\theta_2$. It then follows that the
tightest upper bound on the sum-capacity, denoted as $C_{\rm UB}$,
can be obtained as
\begin{equation}\label{eq:sum rate UB}
C_{\rm UB}=\min_{\kappa_{21}+\kappa_{12}=1} \max_{P_{21}+P_{12}\leq
P_R}C_{21}(\kappa_{21},P_{21})+C_{12}(\kappa_{12},P_{12}).
\end{equation}

Unfortunately, there is in general no closed-form solution of
$C_{\rm UB}$, and thus numerical search over all the feasible values
of $\kappa_{21},\kappa_{12},P_{21}$, and $P_{12}$ is needed to
obtain $C_{\rm UB}$. Since $C_{21}(\kappa_{21},P_{21})$ and
$C_{12}(\kappa_{12},P_{12})$ are increasing functions of $P_{21}$
and $P_{12}$, respectively, a simple upper bound on the sum-capacity
(less tighter than $C_{\rm UB}$) can be obtained from (\ref{eq:C21})
and (\ref{eq:C12}) with $\kappa_{21}=\kappa_{12}=1/2$ and
$P_{21}=P_{12}=P_R$ as follows:
\begin{align}\label{eq:sum rate UB 2}
C_{\rm UB}^{(0)}=&~\frac{1}{2}\log_{2} \left (1 +
\frac{\theta_2p_2}{1+\frac{(\theta_2/\theta_1)p_2}{P_R}+\frac{1}{2\theta_1P_R}}\right)
\nonumber \\ &~ +\frac{1}{2}\log_{2} \left (1 +
\frac{\theta_1p_1}{1+\frac{(\theta_1/\theta_2)p_1}{P_R}+\frac{1}{2\theta_2P_R}}\right).
\end{align}

\begin{remark}\label{remark:sum rate upper bound}
Note that $C_{\rm UB}^{(0)}$ given in (\ref{eq:sum rate UB 2}) can
be used as $\bar{R}_{\rm sum}$ for Algorithm \ref{algorithm:SR max}
in Section \ref{sec:capacity region}. Since $C_{\rm UB}^{(0)}$  is
obtained without any constraint on rate allocations among $r_{21}$
and $r_{12}$, it is a valid upper bound on the achievable sum-rate
regardless of the rate-profile vector $\mv{\alpha}$.
\end{remark}

Next, we derive the lower bounds on the sum-rates achievable by the
proposed suboptimal relay beamforming schemes, MRR-MRT and ZFR-ZFT,
denoted as $R_{\rm LB}^{\rm MR}$ and $R_{\rm LB}^{\rm ZF}$,
respectively. Since the rate lower bound is of interest, we assume
here $a_{\rm x}=b_{\rm x}$, where ${\rm x=MR}$ in (\ref{eq:MRR MRT})
or ${\rm ZF}$ in (\ref{eq:ZFR ZFT}). For conciseness, define
$\rho=\frac{\left|\mv{h}_1^H\mv{h}_2\right|^2}{\theta_1\theta_2}$ as
the correlation coefficient between $\mv{h}_1$ and $\mv{h}_2$. Then,
the following lemmas are obtained.
\begin{lemma}\label{lemma:rate LB MR}
With the MRR-MRT relay beamforming scheme, the achievable sum-rate
of TWRC is lower-bounded by $R_{\rm LB}^{\rm MR}$ given in
(\ref{eq:MR sum rate}) (see next page).
\end{lemma}
\begin{proof}
Please refer to Appendix \ref{appendix:proof MR sum rate}.
\end{proof}

\begin{lemma}\label{lemma:rate LB ZF}
With the ZFR-ZFT relay beamforming scheme, the achievable sum-rate
of TWRC is lower-bounded by $R_{\rm LB}^{\rm ZF}$ given in
(\ref{eq:ZF sum rate}) (see next page).
\end{lemma}

\begin{proof}
Please refer to Appendix \ref{appendix:proof ZF sum rate}.
\end{proof}

\begin{figure*}
\begin{eqnarray}\label{eq:MR sum rate}
R_{\rm LB}^{\rm MR} = \frac{1}{2}\log_{2} \bigg ( 1 +
\frac{\theta_2p_2}{\left(1 + \frac{p_1 +
(\theta_2/\theta_1)p_2}{P_R}\right )\frac{1 + 3 \rho}{(1+\rho)^{2}}
+ \frac{2}{\theta_1(1+\rho) P_R }} \bigg) + \frac{1}{2}\log_{2}
\bigg ( 1 + \frac{\theta_1p_1}{\left (1 +
\frac{(\theta_1/\theta_2)p_1+ p_2}{P_R} \right ) \frac{1 + 3
\rho}{(1+\rho)^{2}} + \frac{2}{ \theta_2(1+\rho) P_R}} \bigg ).
\end{eqnarray}
\begin{eqnarray}\label{eq:ZF sum rate}
R_{\rm LB}^{\rm ZF} &=& \log_{2} \bigg ( 1 +
\frac{2p_1p_2}{\frac{\theta_1+\theta_2}{\theta_1\theta_2(1-\rho)}\left(1+\frac{p_1}{P_R}+\frac{p_2}{P_R}\right)\left(\max(p_1,p_2)
+\frac{\theta_1+\theta_2}{\theta_1\theta_2(1-\rho)}\frac{p_1+p_2}{P_R+p_1+p_2}\right)}\bigg).
\end{eqnarray}
\vspace{-0.2in}
\end{figure*}

\subsection{Asymptotic Results}

Since the main advantage of TWRC over OWRC is to recover the loss of
spectral efficiency due to half-duplex transmissions (see, e.g.,
\cite{Rankov05}-\cite{Tarokh07}), it is important to examine the
achievable sum-rate in TWRC at asymptotically high SNR. In the
following theorem, asymptotic results on various upper and lower
rate bounds in (\ref{eq:sum rate UB 2}), (\ref{eq:MR sum rate}), and
(\ref{eq:ZF sum rate}) are presented.
\begin{theorem}\label{theorem:sum rate high SNR}
Let $p_1$, $p_2$, and $P_R$ all go to infinity with fixed
$\frac{P_R}{p_1}=K_1$ and $\frac{P_R}{p_2}=K_2$. Then, $C_{\rm
UB}^{(0)}$, $R_{\rm LB}^{\rm MR}$, and $R_{\rm LB}^{\rm ZF}$
converge to the values given in (\ref{eq:asymp C UB}),
(\ref{eq:asymp R MR LB}), and (\ref{eq:asymp R ZF LB}), respectively
(see next page).
\end{theorem}

\begin{figure*}
\begin{eqnarray}
C_{\rm UB}^{(0)}&=&
\log_2(P_R)+\frac{1}{2}\log_2\left(\frac{\theta_1\theta_2}{(K_2+\theta_2/\theta_1)(K_1+\theta_1/\theta_2)}
\right) + o(1) \label{eq:asymp C UB}\\
R_{\rm LB}^{\rm
MR}&=&\log_2(P_R)+\frac{1}{2}\log_2\left(\frac{\theta_1\theta_2}{(K_2+K_1/K_2+\theta_2/\theta_1)(K_1+K_2/K_1+\theta_1/\theta_2)\frac{(1
+ 3 \rho)^2}{(1+\rho)^{4}}}
\right) + o(1) \label{eq:asymp R MR LB}\\
R_{\rm LB}^{\rm
ZF}&=&\log_2(P_R)+\log_2\left(\frac{\theta_1\theta_2}{\left(1+\max(K_1,K_2)+\max\left(K_1/K_2,K_2/K_1\right)\right)\frac{\theta_1+\theta_2}{2(1-\rho)}}
\right) + o(1) \label{eq:asymp R ZF LB}.
\end{eqnarray}
\end{figure*}

It is observed from Theorem \ref{theorem:sum rate high SNR} that at
high SNR both MRR-MRT and ZFR-ZFT (provided that $\rho<1$)
asymptotically achieve the same sum-rate pre-log factor (sum-rate
normalized by $\log_2P_R$ as $P_R\rightarrow \infty$) as that of the
sum-capacity upper bound, $C_{\rm UB}^{(0)}$. However, they may have
different rate gaps from $C_{\rm UB}^{(0)}$, which are constants
independent of $P_R$. In order to gain more insights on the limiting
sum-rate losses of suboptimal beamforming schemes, in the following
corollary, we compare the difference between $C_{\rm UB}^{(0)}$ and
$R_{\rm LB}^{\rm MR}$ or $R_{\rm LB}^{\rm ZF}$ at asymptotically
high SNR in a ``symmetric'' TWRC with equal channel gains, i.e.,
$\|\mv{h}_1\|^2=\|\mv{h}_2\|^2=\theta$, and equal source and relay
transmit powers, i.e., $p_1=p_2=P_R$. In this case, we can obtain a
tighter upper bound on the sum-capacity than $C_{\rm UB}^{(0)}$ as
follows. For the symmetric TWRC, with $\kappa_{12}=\kappa_{21}=1/2$,
it can be easily verified that the maximization over $P_{12}$ and
$P_{21}$ in (\ref{eq:sum rate UB}) is achieved when
$P_{12}=P_{21}=P_R/2$ and as a result a tighter upper bound over
$C_{\rm UB}^{(0)}$ for the symmetric TWRC is obtained as
\begin{equation}\label{eq:sum rate UB 3}
C_{\rm UB}^{(S)}=\log_{2} \left (1 + \frac{\theta
P_R}{3+\frac{1}{\theta P_R}}\right).
\end{equation}

\begin{corollary}\label{corol:high SNR}
At asymptotically high SNR, under the assumptions that
$\theta_1=\theta_2$ and $K_1=K_2=1$, we have
\begin{eqnarray}
C_{\rm UB}^{(S)}-R_{\rm LB}^{\rm
MR}&=&\log_2\left(\frac{1+3\rho}{(1+\rho)^2}\right) \label{eq:rate gap MR}\\
C_{\rm UB}^{(S)}-R_{\rm LB}^{\rm ZF}&=&\log_2\left(
\frac{1}{1-\rho}\right). \label{eq:rate gap ZF}
\end{eqnarray}
\end{corollary}

It is noted that for $0\leq \rho\leq 1$,
$\frac{1+3\rho}{(1+\rho)^2}$ has the minimum value equal to 0 at
$\rho=0$ or $1$, and the maximum value equal to $9/8$ at $\rho=1/3$.
Therefore, from (\ref{eq:rate gap MR}), it follows that the sum-rate
loss of MRR-MRT from the sum-capacity is at most
$\log_2(9/8)\approx0.1699$ bits/complex dimension at asymptotically
high SNR. On the other hand, it is observed from (\ref{eq:rate gap
ZF}) that the sum-rate loss resulted by ZFR-ZFT increases with
$\rho$, or when $\mv{h}_1$ and $\mv{h}_2$ become more correlated.
This is intuitively correct, since with the increasing channel
correlation, more SNR loss will be incurred to separate the signals
from/to S1 and S2 at R by ZF-based receive/transmit beamforming.
Also note that for MRR-MRT, at $\rho=0$ or $1$, the sum-rate loss is
zero, which is consistent with Lemma \ref{lemma:optimality MR},
while the sum-rate loss is zero for ZFR-ZFT at $\rho=0$, which is
consistent with Lemma \ref{lemma:optimality ZF}.

\section{Numerical Results} \label{sec:numerical results}

In this section, we present numerical results on the achievable
rates of various beamforming schemes considered in this paper, and
compare them with those of other existing schemes in the literature.
For convenience, we assume that $\mv{h}_1$ is a randomly generated
CSCG vector $\sim\mathcal{CN}(\mv{0},\mv{I})$, and $\mv{h}_1$ is
normalized by its own vector norm such that $\|\mv{h}_1\|=1$. We
then generate $\mv{h}_2$ according to
$\mv{h}_2=\sqrt{\rho}\mv{h}_1+\sqrt{1-\rho}\mv{h}_w$, where
$\mv{h}_w$ is also a normalized CSCG random vector, $\|\mv{h}_w\|=1$
and $\mv{h}_1^H\mv{h}_w=0$. Thereby, it can be easily verified that
$\|\mv{h}_2\|=1$ and $\|\mv{h}_1^H\mv{h}_2\|^2=\rho$. It is assumed
that $M=4$ in this section.

\subsection{Capacity Region of ANC-Based TWRC}

\begin{figure}
\centering{
 \epsfxsize=3.4in
    \leavevmode{\epsfbox{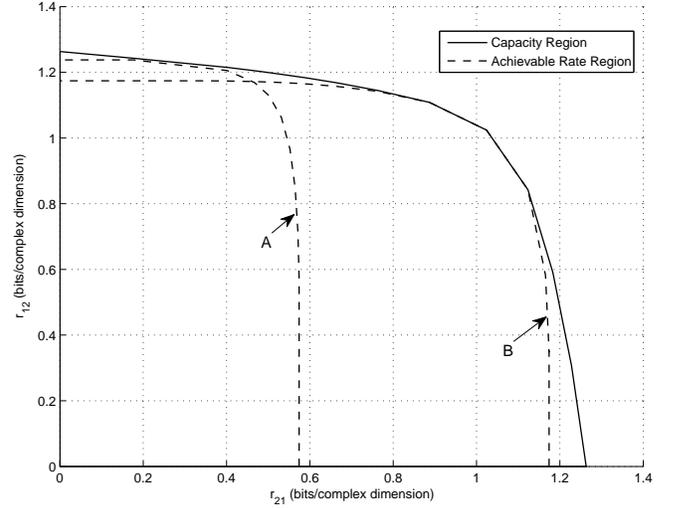}} }
    \caption{Capacity region of the ANC-based TWRC with $M=4$,
$P_1=P_2=P_R=10$, and $\rho=0.5$. Note that the two rate regions
enclosed by the dashed lines are example achievable rate regions
$\mathcal{R}(p_1,p_2,P_R)$'s defined in (\ref{eq:rate region}), each
with some fixed $p_1$ and $p_2$, $p_1\leq P_1$ and $p_2\leq P_2$.
The achievable rate region denoted by A corresponds to $p_1=P_1$ and
$p_2<P_2$, while that denoted by B corresponds to $p_1=P_1$ and
$p_2=P_2$.}\label{fig:capacity region} \vspace{-0.1in}
\end{figure}

Fig. \ref{fig:capacity region} shows the capacity region,
$\mathcal{C}(P_1,P_2,P_R)$ defined in (\ref{eq:capacity region}),
for the ANC-based TWRC with $P_1=P_2=P_R=10$, and $\rho=0.5$. It is
observed that $\mathcal{C}(P_1,P_2,P_R)$ is symmetric over $r_{12}$
and $r_{21}$ in this case. Notice that boundary rate-pairs of
$\mathcal{C}(P_1,P_2,P_R)$ are resulted by the union over those of
achievable rate regions, $\mathcal{R}(p_1,p_2,P_R)$'s defined in
(\ref{eq:rate region}), with different values of $p_1$ and $p_2$,
$0\leq p_1 \leq P_1$ and $0\leq p_2 \leq P_2$. Boundary rate-pairs
of each constituting $\mathcal{R}(p_1,p_2,P_R)$ are obtained by
solving problem (\ref{eq:SRMax}) using Algorithm \ref{algorithm:SR
max} with different rate-profile vectors $\mv{\alpha}$'s. It is
observed that when $p_1=0$ and $p_2=P_2$, $\mathcal{R}(p_1,p_2,P_R)$
collapses into the horizontal rate axis of $r_{21}$, and the maximum
value of $r_{21}$ becomes the capacity of the OWRC with S2
transmitting to S1 via R. Similarly, $\mathcal{R}(P_1,0,P_R)$
collapses into the vertical rate axis of $r_{12}$, and the maximum
value of $r_{12}$ becomes the capacity of the OWRC with S1
transmitting to S2 via R.

\subsection{Achievable Rates of Suboptimal Beamforming Schemes}

\begin{figure}
\centering{
 \epsfxsize=3.4in
    \leavevmode{\epsfbox{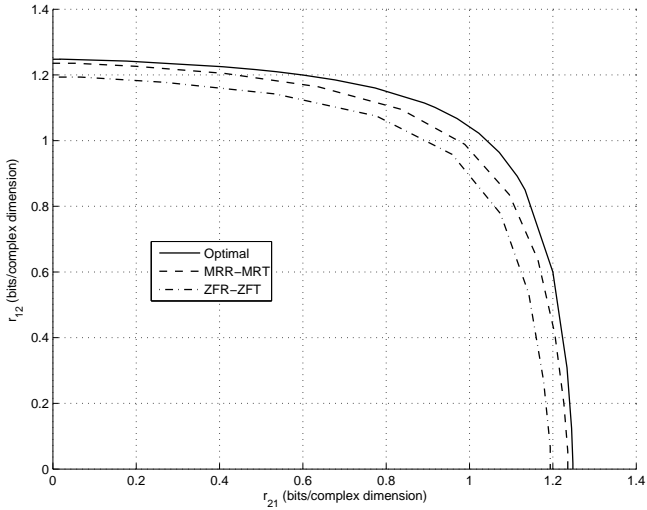}} }
\caption{Achievable rate regions of the ANC-based TWRC with $M=4$,
$p_1=p_2=10$, $P_R=10$, and $\rho=0.1$.}\label{fig:rate comp low
cor} \vspace{-0.1in}
\end{figure}

\begin{figure}
\centering{
 \epsfxsize=3.4in
    \leavevmode{\epsfbox{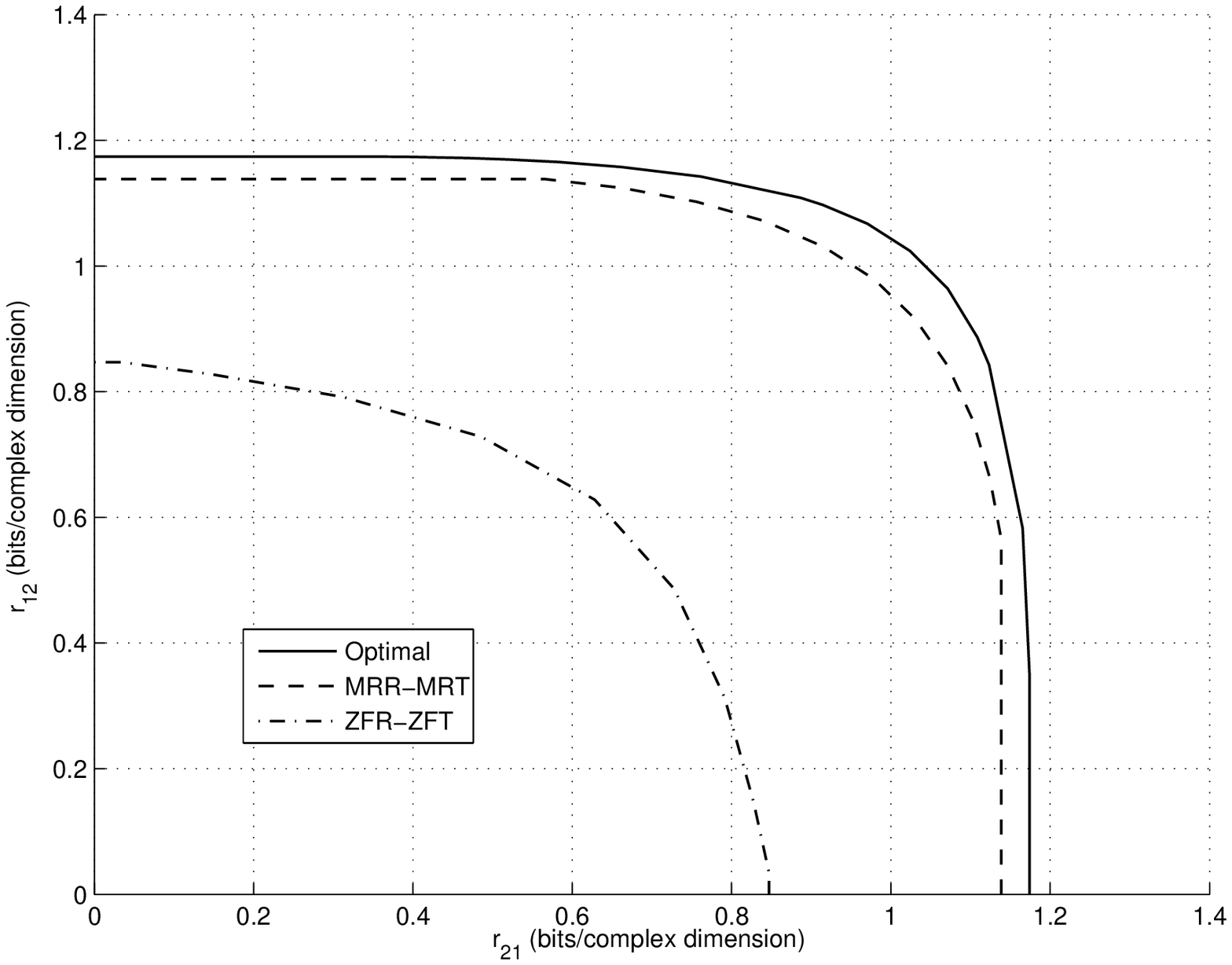}} }
\caption{Achievable rate region of the ANC-based TWRC with $M=4$,
$p_1=p_2=10$, $P_R=10$, and $\rho=0.5$.}\label{fig:rate comp
moderate cor} \vspace{-0.1in}
\end{figure}

\begin{figure}
\centering{
 \epsfxsize=3.4in
    \leavevmode{\epsfbox{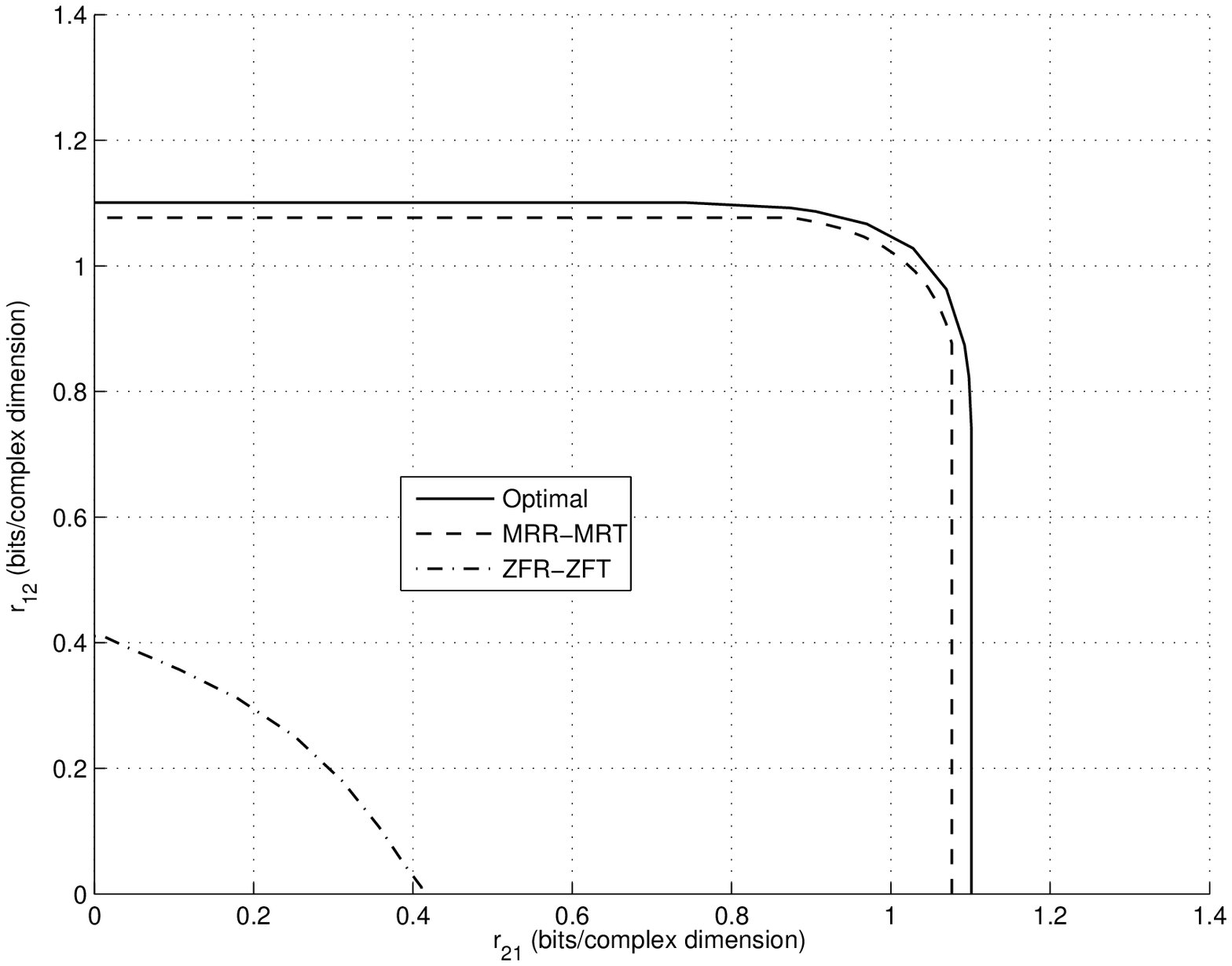}} }
\caption{Achievable rate regions of the ANC-based TWRC with $M=4$,
$p_1=p_2=10$, $P_R=10$, and $\rho=0.8$.}\label{fig:rate comp high
cor}\vspace{-0.1in}
\end{figure}

\begin{figure}
\centering{
 \epsfxsize=3.4in
    \leavevmode{\epsfbox{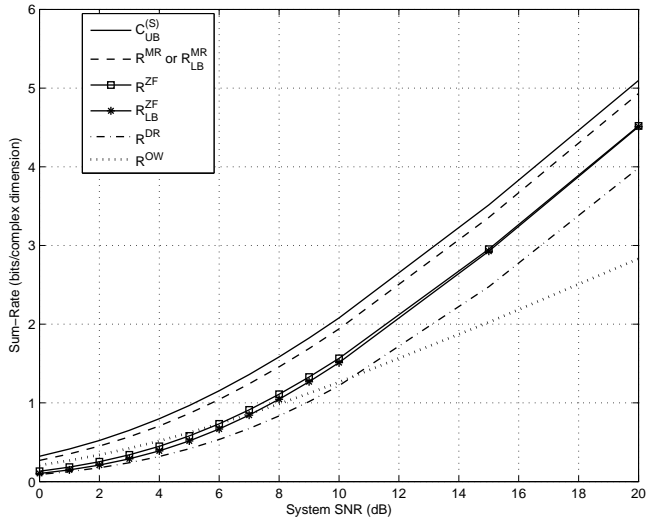}} }
\caption{Sum-rate versus system SNR for the ANC-based TWRC with
$M=4$ and $\rho=1/3$.}\label{fig:sum rate} \vspace{-0.2in}
\end{figure}

Next, we examine the achievable rates of the proposed suboptimal
relay beamforming schemes. Figs. \ref{fig:rate comp low cor},
\ref{fig:rate comp moderate cor}, and \ref{fig:rate comp high cor}
show the achievable rate region, $\mathcal{R}(p_1,p_2,P_R)$, for the
TWRC with different values of $\rho$, $\rho=0.1,0.5$, and 0.8,
respectively. It is assumed that transmit powers at S1 and S2 are
fixed as $p_1=p_2=10$, and the relay transmit power constraint is
$P_R=10$. Three relay beamforming schemes are compared in each of
these figures, which are the optimal scheme (Algorithm
\ref{algorithm:SR max}), the MRR-MRT scheme (\ref{eq:MRR MRT}), and
the ZFR-ZFT scheme (\ref{eq:ZFR ZFT}). Note that boundary rate-pairs
of $\mathcal{R}(p_1,p_2,P_R)$ corresponding to MRR-MRT are obtained
by changing different ratios between $a_{\rm MR}$ and $b_{\rm MR}$
in (\ref{eq:MRR MRT}). $\mathcal{R}(p_1,p_2,P_R)$ for ZFR-ZFT is
obtained in a similar way. It is observed that the achievable rate
region by MRR-MRT is very close to that with the optimal scheme when
the channel correlation coefficient $\rho$ is either small or large,
which is in accord with Lemma \ref{lemma:optimality MR}. Even for
moderate values of $\rho$, e.g., $\rho=0.5$, the rate loss of
MRR-MRT is observed to be negligible, suggesting that MRR-MRT in
fact performs very close to the optimal scheme under different
channel conditions. In contrast, ZFR-ZFT performs close to the
optimal scheme when $\rho$ is small, which is in accord with Lemma
\ref{lemma:optimality ZF}. This is due to the fact that when
$\mv{h}_1$ and $\mv{h}_2$ are sufficiently decorrelated, ZF-based
recieve/transmit beamforming at R is able to suppress the UL/DL
interference between S1 and S2 with small SNR losses. However, as
$\rho$ increases, it is observed that the achievable rates of
ZFR-ZFT degrade significantly as compared to those of the optimal
scheme or MRR-MRT.

In Fig. \ref{fig:sum rate}, we show the achievable sum-rate of TWRC
with $\rho=1/3$ versus the ``system'' SNR. Under the assumption that
transmit powers at S1, S2, and R are all equal, i.e., $p_1=p_2=P_R$,
due to the unit-norm channels and unit-variance noises, the system
SNR is conveniently set equal to $P_R$. Various sum-rate bounds
presented in this paper are shown, including $C_{\rm UB}^{(S)}$ in
(\ref{eq:sum rate UB 3}), $R_{\rm LB}^{\rm MR}$ in (\ref{eq:MR sum
rate}), and $R_{\rm LB}^{\rm ZF}$ in (\ref{eq:ZF sum rate}). In
addition, the actual achievable sum-rates of MRR-MRT and ZFR-ZFT,
denoted as $R^{\rm MR}$ and $R^{\rm ZF}$, respectively, are also
shown for comparison. Note that due to the channel symmetry, $a_{\rm
MR}$ and $b_{\rm MR}$ in (\ref{eq:MRR MRT}) should be equal to
maximize $R^{\rm MR}$; thus, from the derivations in Appendix
\ref{appendix:proof MR sum rate} it follows that $R^{\rm
MR}$=$R_{\rm LB}^{\rm MR}$ for MRR-MRT. On the other hand, for
ZFR-ZFT, $a_{\rm ZF}$ and $b_{\rm ZF}$ in (\ref{eq:ZFR ZFT}) should
also be equal to maximize $R^{\rm ZF}$ in this symmetric-channel
case; however, from Appendix \ref{appendix:proof ZF sum rate} it
follows that even with $a_{\rm ZF}=b_{\rm ZF}$, $R_{\rm LB}^{\rm
ZF}<R^{\rm ZF}$ in general, where $R^{\rm ZF}$ can be obtained from
the RHS of (\ref{eq:inequality ZF 0}). We also show the sum-rates of
the following two heuristic schemes: (1) {\it Direct relaying},
where the relay beamforming matrix is in the form of
$\mv{A}=\zeta\mv{I}$, with $\zeta$ being a constant determined by
$P_R$; (2) {\it One-way alternative relaying}, where four time-slots
are used for one round of information exchange between S1 and S2,
with two for S2 transmitting to S1 via R, and the other two for S1
to S2 via R, and the corresponding optimal relay beamforming
matrices are in the form of $\mv{A}_{21}=\psi\mv{h}_1^*\mv{h}_2^H$
and $\mv{A}_{12}=\psi\mv{h}_2^*\mv{h}_1^H$, respectively, with
$\psi$ determined by $P_R$ \cite{Hua07}, \cite{Vidal07}. We denote
$R^{\rm DR}$ and $R^{\rm OW}$ as the achievable sum-rates of these
two schemes, respectively.

It is observed in Fig. \ref{fig:sum rate} that at asymptotically
high SNR, the sum-rate of MRR-MRT converges to the sum-capacity
upper bound with a constant gap of $0.1699$ bits/complex dimension,
while ZFR-ZFT has a sum-rate gap of $\log_2(1/(1-1/3))=0.5850$
bits/complex dimension. The above observations agree with Corollary
\ref{corol:high SNR}. It is also observed that the lower bound on
the sum-rate by ZFR-ZFT, $R_{\rm LB}^{\rm ZF}$, is very tight at all
SNR values. Notice that $R^{\rm DR}$ and $R^{\rm OW}$ both have
significant gaps from $R^{\rm MR}$ at asymptotically high SNR, since
the former has no beamforming gain at R, and the latter roughly
incurs a loss of half the spectral efficiency due to alternative
relaying.

\subsection{Comparison with DF-Based TWRC}

At last, we compare the capacity region of ANC/AF-based TWRC derived
in this paper with that of DF-based TWRC recently reported in
\cite{Boche08b}, for the same physical TWRC. In order to
differentiate the above two capacity regions, we denote the former
as $\mathcal{C}_{\rm AF}$ and the latter as $\mathcal{C}_{\rm DF}$.
Note that with DF relay operation, R first decodes both messages
from S1 and S2 as in the conventional Gaussian multiple-access
channel (MAC) during the first time-slot; R then re-encodes the
decoded messages jointly into a new message, and transmits it over
the broadcast channel (BC) to both S1 and S2 during the second
time-slot. Each of S1 and S2 decodes the message of the other from
the received signal given the side information on its own previously
transmitted message (in the first time-slot). The achievable rates
of S2 and S1 during the first MAC phase can be expressed as
\cite{Cover}
\begin{align}\label{eq:capacity region MAC}
& \mathcal{C}^{\rm MAC}_{\rm DF}(P_1,P_2) \triangleq
\bigg\{(r_{21},r_{12}): r_{21}\leq
\log_2\left|\mv{I}+P_2\mv{h}_2\mv{h}_2^H\right|, \nonumber
\\  & ~~~~~~~~~~~~~~~~~~~~~~~ r_{12}\leq
\log_2\left|\mv{I}+P_1\mv{h}_1\mv{h}_1^H\right|, \nonumber
\\  & ~~~~~~~~~ r_{21}+r_{12}\leq
\log_2\left|\mv{I}+P_1\mv{h}_1\mv{h}_1^H+P_2\mv{h}_2\mv{h}_2^H
\right| \bigg\}.
\end{align}

The maximum achievable rate-pairs during the second BC phase can be
expressed as \cite{Boche08b}
\begin{align}\label{eq:capacity region BC}
&\mathcal{C}^{\rm BC}_{\rm DF}(P_R)\triangleq \bigcup_{\mv{S}_R:
\mv{S}_R\succeq 0, \mathtt{tr}(\mv{S}_R)\leq
P_R}\bigg\{(r_{21},r_{12}): \nonumber
\\ & r_{21}\leq \log_2\left(1+\mv{h}_1^T\mv{S}_R\mv{h}_1^*\right),
r_{12}\leq \log_2\left(1+\mv{h}_2^T\mv{S}_R\mv{h}_2^*\right) \bigg\}
\end{align}
where $\mv{S}_R$ is the transmit signal covariance matrix at R. Note
that in order to obtain $\mathcal{C}^{\rm BC}_{\rm DF}(P_R)$ in
(\ref{eq:capacity region BC}), we need to solve a sequence of
optimization (WSRMax) problems expressed below with different
nonnegative rate weights $w_{21}$ and $w_{12}$.
\begin{align}\label{eq:WSRMax DF}
\mathop{\mathtt{Max.}}_{\mv{S}_R} &~~
w_{21} \log_2\left(1+\mv{h}_1^T\mv{S}_R\mv{h}_1^*\right) + w_{12}\log_2\left(1+\mv{h}_2^T\mv{S}_R\mv{h}_2^*\right)\nonumber \\
\mathtt{s.t.} &~~ \mathtt{tr}(\mv{S}_R)\leq P_R, \ \mv{S}_R\succeq
0.
\end{align}
Since the above problem is convex, it can be solved by standard
convex optimization techniques, e.g., the interior-point method
\cite{Boydbook}. Unlike the AF relay operation, DF relay operation
allows different time allocations between the MAC and BC time-slots.
Let $\tau$ and $1-\tau$ denote the percentages of the total time
allocated to MAC phase and BC phase, respectively. Then, combining
both MAC and BC phases yields the capacity region for DF-based TWRC,
$\mathcal{C}_{\rm DF}(P_1,P_2,P_R)$, expressed as
\begin{eqnarray}\label{eq:capacity region DF}
\bigcup_{\tau:0\leq\tau\leq 1} \left( \tau\cdot\mathcal{C}^{\rm
MAC}_{\rm DF}(P_1,P_2)\bigcap (1-\tau)\cdot\mathcal{C}^{\rm BC}_{\rm
DF}(P_R)\right).
\end{eqnarray}

In Figs. \ref{fig:DF AF comp HCor} and \ref{fig:DF AF comp LCor}, we
show $\frac{1}{2}\mathcal{C}^{\rm MAC}_{\rm DF}$,
$\frac{1}{2}\mathcal{C}^{\rm BC}_{\rm DF}$, $\mathcal{C}_{\rm DF}$,
and $\mathcal{C}_{\rm AF}$ for $\rho=0.95$ and $0.8$, respectively.
It is assumed that $P_1=P_2=P_R=100$. Note that in each figure,
$\mathcal{C}_{\rm DF}$ can be visualized as the union of rate
regions, each of which corresponds to the intersection of
$\frac{1}{2}\mathcal{C}^{\rm MAC}_{\rm DF}$ and
$\frac{1}{2}\mathcal{C}^{\rm BC}_{\rm DF}$ after they are properly
scaled by $2\tau$ and $2(1-\tau)$, respectively, for a particular
value of $\tau$. It is observed that the DF-based TWRC in general
has a larger capacity region over the AF-based counterpart.
Furthermore, it is observed that this capacity gain enlarges as
$\rho$ decreases, i.e., the channels $\mv{h}_1$ and $\mv{h}_2$
become more weakly correlated. This is mainly due to the fact that
when the UL channels become less correlated, R is more capable of
decoding the messages from S1 and S2 during the MAC phase and as a
result, $\frac{1}{2}\mathcal{C}^{\rm MAC}_{\rm DF}$ is observed to
get enlarged as $\rho$ decreases. If we want to draw a more fair
comparison between AF- and DF-based TWRCs with the same energy
consumption, we may assume that for the DF-based TWRC,
equal-duration time-slots are assigned to the MAC and BC phases,
i.e., $\tau=1/2$, the same as the AF case. As such, since for both
$\rho=0.95$ and $0.8$, $\frac{1}{2}\mathcal{C}^{\rm MAC}_{\rm DF}$
appears as a subset of $\frac{1}{2}\mathcal{C}^{\rm BC}_{\rm DF}$,
it concludes that the capacity region for DF-based TWRC with the
fixed $\tau=1/2$ is simply $\frac{1}{2}\mathcal{C}^{\rm MAC}_{\rm
DF}$. Interestingly, it is observed that $\mathcal{C}_{\rm AF}$
improves over $\frac{1}{2}\mathcal{C}^{\rm MAC}_{\rm DF}$ when
$\rho=0.95$ in the region where the values of $r_{21}$ and $r_{12}$
are close to each other. Notice that in this region the sum-capacity
in the AF case is achieved. Since DF relaying incurs larger
complexity for encoding/decoding at R as compared with AF relaying,
AF relaying may be a more suitable solution in practice where strong
channel correlation is encountered.\footnote{Note that in the
extreme case of $\rho=1$, the multi-antenna TWRC becomes equivalent
to the single-antenna TWRC studied in \cite{ANC}.}  However, in the
case of $\rho=0.8$, it is observed that the capacity improvement of
$\mathcal{C}_{\rm AF}$ over $\frac{1}{2}\mathcal{C}^{\rm MAC}_{\rm
DF}$ diminishes.

\begin{figure}
\centering{
 \epsfxsize=3.4in
    \leavevmode{\epsfbox{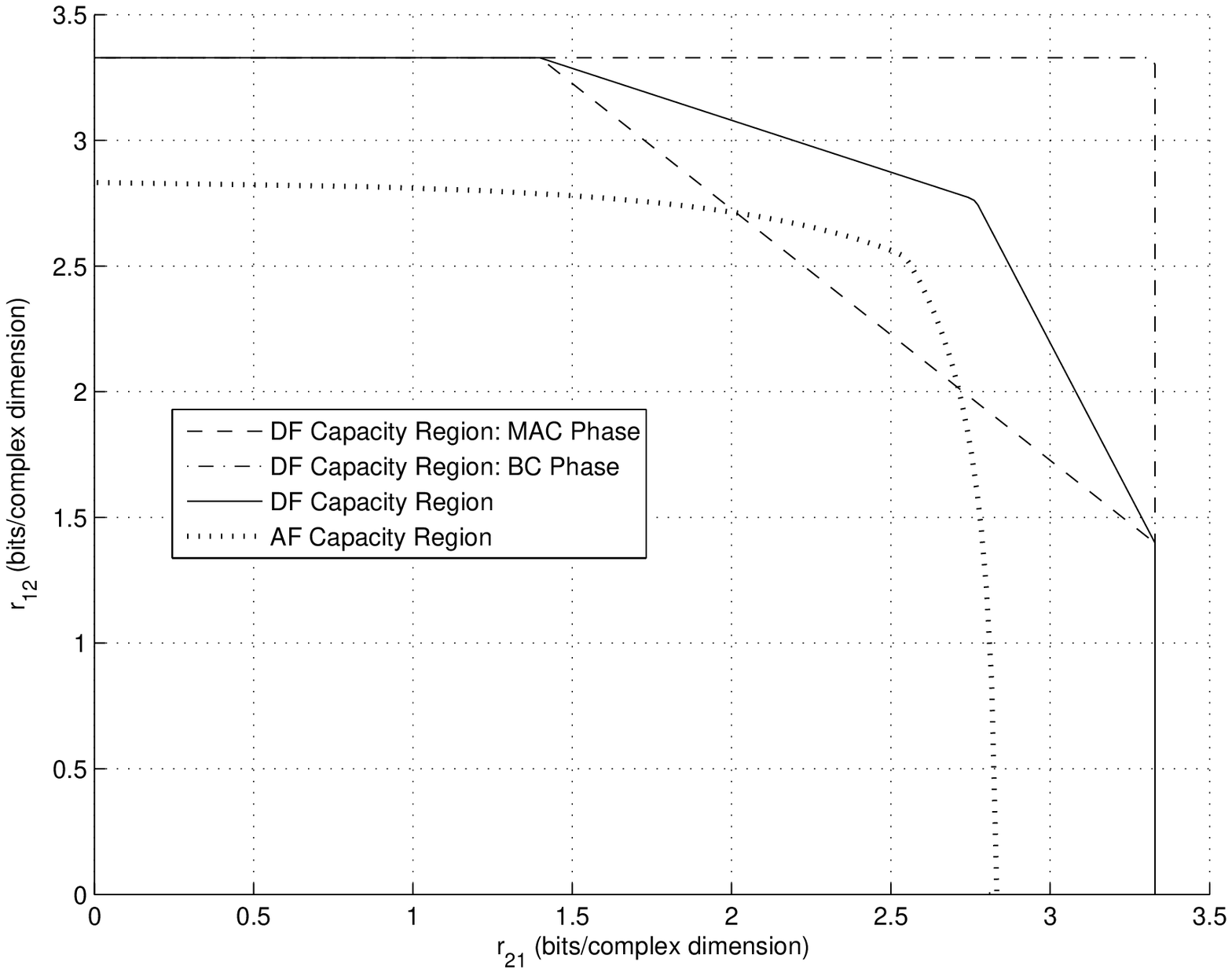}} }
\caption{Comparison of capacity region for ANC/AF-based versus
DF-based TWRC with $M=4$, $P_1=P_2=P_R=100$, and
$\rho=0.95$.}\label{fig:DF AF comp HCor} \vspace{-0.1in}
\end{figure}

\begin{figure}
\centering{
 \epsfxsize=3.4in
    \leavevmode{\epsfbox{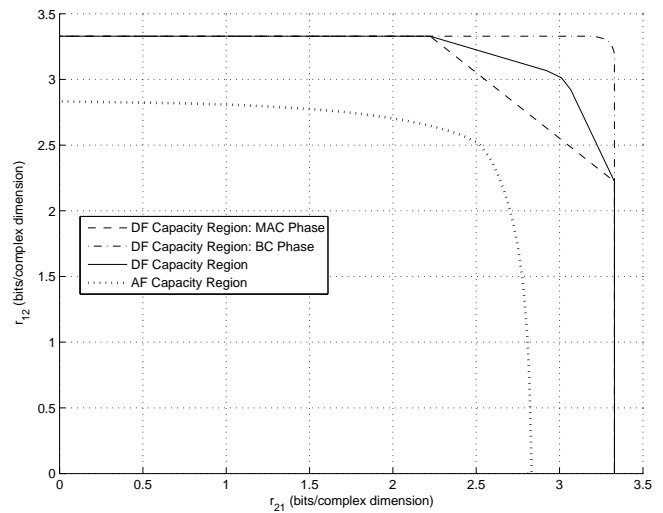}} }
\caption{Comparison of capacity region for ANC/AF-based versus
DF-based TWRC with $M=4$, $P_1=P_2=P_R=100$, and
$\rho=0.8$.}\label{fig:DF AF comp LCor} \vspace{-0.1in}
\end{figure}

\section{Conclusion and Future Work}\label{sec:conclusions}

This paper studied the fundamental capacity limits of ANC/AF-based
TWRC with multi-antennas at the relay. It was shown that the
standard method to characterize the capacity region via WSRMax is
not directly applicable to ANC-based TWRC due to the non-convexity
of the optimization problem. Therefore, we proposed an alternative
method to characterize the capacity region of TWRC by applying the
idea of rate profile. As a byproduct, we also provided the solution
for the relay power minimization problem under given SNR constraints
at the receivers. Due to the bidirectional transmission as well as
the self-interference cancelation by ANC, we found that the design
of relay beamforming in TWRC differs very much from the conventional
designs for the OWRC or the UL/DL beamforming in the traditional
cellular network. We presented the general form of the optimal relay
beamforming structure in TWRC, as well as two low-complexity
suboptimal schemes, namely, MRR-MRT and ZFR-ZFT. It was shown that
ZFR-ZFT with the objective of suppressing the UL and DL
interferences between S1 and S2 may not perform well in the case of
strong channel correlation, while MRR-MRT with the objective of
maximizing the total forwarded signal power from R to S1 and S2
achieves sum-rates and rate regions close to the optimal ones under
various SNR and channel conditions. This suggests that MRR-MRT can
be a good solution from an implementation viewpoint. It was also
shown that the ANC/AF-based TWRC can have a capacity gain over the
DF-based TWRC for sufficiently large channel correlations and equal
MAC and BC time-durations.

Future work beyond this paper may include the joint design of source
and relay beamforming when each source is also equipped with
multi-antennas, the relay beamforming design for more than one
source-pairs with different combined unicast/multicast
transmissions, and the design of a hybrid AF/DF scheme that probably
improves the performances of both AF- and DF-based TWRCs. In
addition, the study of estimate-and-forward (EF) relay operations
for the multi-antenna TWRC is also appealing.

\appendices

\section{Proof of Theorem \ref{theorem:optimal beamforming}}\label{appendix:proof beamforming
optimality}

Without loss of generality, we can express $\mv{A}$ as
\begin{eqnarray}
\mv{A}&=&[\mv{U}^*, (\mv{U}^{\bot})^*] \left[\begin{array}{ll}
\mv{B} & \mv{C} \\ \mv{D} & \mv{E} \end{array} \right] [\mv{U},
\mv{U}^{\bot}]^H \\
&=&
\mv{U}^*\mv{BU}^H+\mv{U}^*\mv{C}(\mv{U}^{\bot})^H+(\mv{U}^{\bot})^*\mv{DU}^H
\nonumber \\ && +(\mv{U}^{\bot})^*\mv{E}(\mv{U}^{\bot})^H
\label{eq:matrix sum}
\end{eqnarray}
where $\mv{U}^{\bot}\in\mathbb{C}^{M\times (M-2)}$,
$\mv{U}^{\bot}(\mv{U}^{\bot})^H=\mv{I}-\mv{U}\mv{U}^H$, and
$\mv{B}$, $\mv{C}$, $\mv{D}$, and $\mv{E}$ are complex matrices of
size $2\times 2$, $2\times (M-2)$, $(M-2)\times 2$, and $(M-2)\times
(M-2)$, respectively. First, it can be shown that in (\ref{eq:rate
21}),
$|\mv{h}_1^T\mv{A}\mv{h}_2|^2=|\mv{h}_1^T\mv{U}^*\mv{B}\mv{U}^H\mv{h}_2|^2$,
and
$\|\mv{A}^H\mv{h}_1^*\|^2=\|\mv{B}^H\mv{U}^T\mv{h}_1^*\|^2+\|\mv{C}^H\mv{U}^T\mv{h}_1^*\|^2\geq
\|\mv{B}^H\mv{U}^T\mv{h}_1^*\|^2$. Thus, it follows that $r_{21}$
does not depend on $\mv{D}$ and $\mv{E}$, and is maximized when
$\mv{C}=\mv{0}$. Similarly, from (\ref{eq:rate 12}), we can show
that $r_{12}$ is also not related to $\mv{D}$ and $\mv{E}$, and is
maximized when $\mv{C}=\mv{0}$. Next, for the relay power constraint
(\ref{eq:rate region}), from (\ref{eq:relay power}) it can be shown
that $p_R$ is minimized when $\mv{C}$, $\mv{D}$, and $\mv{E}$ are
all equal to $\mv{0}$. Since each rate-pair on the boundary of
$\mathcal{R}(p_1,p_2,P_R)$ defined in (\ref{eq:rate region}) must
maximize $r_{21}$ and $r_{12}$ subject to the given $P_R$, it
concludes that all $\mv{C}$, $\mv{D}$, and $\mv{E}$ in the
corresponding $\mv{A}$ should be $\mv{0}$. Thus, from
(\ref{eq:matrix sum}), we conclude that $\mv{A}=\mv{U}^*\mv{BU}^H$.

\section{Proof of Lemma \ref{lemma:zero correlation}}\label{appendix:proof zero correlation}

In the case of  $\mv{h}_1\bot\mv{h}_2$, it can be easily shown that
$\mv{g}_1=[\|\mv{h}_1\|, 0]^T$ and $\mv{g}_2=[0, \|\mv{h}_2\|]^T$.
Let $\mv{B}=\left[{\footnotesize
\begin{array}{cc}
a & c\\
d & b \end{array} }\right]$. Substituting $\mv{g}_1$ and $\mv{g}_2$
into (\ref{eq:rate region 2}) yields
\begin{equation}
r_{21} \leq \frac{1}{2}\log_{2} \left (1 +
\frac{\|\mv{h}_1\|^2\|\mv{h}_2\|^2|c|^2p_2
}{\|\mv{h}_1\|^2(|a|^2+|c|^2)+1}\right)
\end{equation}
\begin{equation}
 r_{12} \leq
\frac{1}{2}\log_{2} \left (1 +
\frac{\|\mv{h}_1\|^2\|\mv{h}_2\|^2|d|^2p_1
}{\|\mv{h}_2\|^2(|b|^2+|d|^2)+1}\right)
\end{equation}
\begin{eqnarray}
\|\mv{h}_1\|^2(|a|^2+|d|^2)+\|\mv{h}_2\|^2(|c|^2+|b|^2)\nonumber
\\ +|a|^2+|b|^2+|c|^2+|d|^2\leq P_R.
\end{eqnarray}

It then follows that $r_{21}$ and $r_{12}$ are maximized along with
the relay transmit power being minimized when $|a|=0$ and $|b|=0$
and, thus, $a=b=0$. Since in the above rate and power expressions
only $|c|^2$ and $|d|^2$ are involved, we can assume w.l.o.g. that
$c\geq 0$ and $d\geq 0$.

\section{Proof of Lemma \ref{lemma:unit correlation}}\label{appendix:proof unit correlation}

In the case of  $\mv{h}_1\parallel \mv{h}_2$, it can be easily shown
that $\mv{g}_1=[\|\mv{h}_1\|, 0]^T$ and $\mv{g}_2=[ \|\mv{h}_2\|,
0]^T$. Let $\mv{B}=\left[{\footnotesize
\begin{array}{cc}
a & c\\
d & b \end{array} }\right]$. Similarly like the proof of Lemma
\ref{lemma:zero correlation} in Appendix \ref{appendix:proof zero
correlation}, by substituting $\mv{g}_1$ and $\mv{g}_2$ into
(\ref{eq:rate region 2}), it follows that $r_{21}$ and $r_{12}$ are
maximized along with the relay transmit power being minimized when
$b=c=d=0$, and  we can assume w.l.o.g. that $a\geq 0$.

\section{Proof of Convergence of  Algorithm \ref{algorithm:SR max}}
\label{appendix:proof algorithm}

In this appendix, we prove that Algorithm \ref{algorithm:SR max}
guarantees the convergence of $r_{\min}$ to the optimal solution of
problem (\ref{eq:SRMax}). First, we show that $r_{\min}$ is a
feasible solution of problem (\ref{eq:SRMax}): Given $R_{\rm
sum}=r_{\min}$, from Algorithm \ref{algorithm:SR max} it is easily
verified that all the three constraints of problem (\ref{eq:SRMax})
are satisfied. Secondly, suppose that there exists another feasible
solution $\hat{r}$ for problem (\ref{eq:SRMax}) such that
$\hat{r}>r_{\min}+\delta_r$ (note that $\delta_r$ can be chosen
arbitrarily small in Algorithm \ref{algorithm:SR max}). However,
this contradicts the fact that $r_{\max}$, $r_{\max}\leq
r_{\min}+\delta_r <\hat{r}$, has been proven in Algorithm
\ref{algorithm:SR max} to be an infeasible solution of problem
(\ref{eq:SRMax}) since the required minimum power, $p_R^{\star}$, is
larger than the given constraint $P_R$ in problem (\ref{eq:SRMax}).
Therefore, by contradiction, it follows that there does not exist
such a feasible solution $\hat{r}$ for problem (\ref{eq:SRMax}).
From the above discussions, it concludes that the feasible solution
$r_{\min}$ is at most $\delta_r$ lower than the optimal solution of
problem (\ref{eq:SRMax}). By letting $\delta_r\rightarrow 0$,
convergence of Algorithm \ref{algorithm:SR max} is thus proved.

\section{Proof of Theorem \ref{theorem:rank_one}}
\label{appendix:proof rank one}

Given $\mv{X}^\star$, first we know that at least one of the two
inequality constraints in~(\ref{eq:PowerMin_SDP}) is active at the
optimal point, i.e., we have either
$\mathtt{tr}(\mv{F}_1\mv{X}^\star)= 1$ or
$\mathtt{tr}(\mv{F}_2\mv{X}^\star)=1$, or both. This fact can be
proved by contradiction: If at $\mv{X}^\star$ both
$\mathtt{tr}(\mv{F}_1\mv{X}^\star)> 1$ and
$\mathtt{tr}(\mv{F}_2\mv{X}^\star)>1$ hold, we could always find a
$t$ with $0<t<1$  such that $\mv{Y}^\star=t\mv{X}^\star$ and
$\min{\left(\mathtt{tr}(\mv{F}_1\mv{Y}^\star),
\mathtt{tr}(\mv{F}_2\mv{Y}^\star)\right)}=1$. We could easily see
that
$\mathtt{tr}(\mv{F}_0\mv{Y}^\star)<\mathtt{tr}(\mv{F}_0\mv{X}^\star)$,
which means that $\mv{X}^\star$ could not be the optimal solution,
i.e., contradiction holds.

From now on, we assume w.l.o.g. that
$\mathtt{tr}(\mv{F}_1\mv{X}^\star)= 1$ such that we have
$\mathtt{tr}((\mv{F}_2-\mv{F}_1)\mv{X}^\star)\ge 0$. To facilitate
the proof for Theorem~\ref{theorem:rank_one}, let us first give the
following lemma, which is based on Lemma~1 given in~\cite{Ye_Zhang},
and the proof also follows a similar way to that in~\cite{Ye_Zhang}
(so it is skipped here).
\begin{lemma}\label{lemma:decomposition}
Given that $\mathtt{tr}((\mv{F}_2-\mv{F}_1)\mv{X}^\star)\ge 0$,
there exists a decomposition for $\mv{X}^\star$ such that
$$\mv{X}^\star=\sum_{i=1}^{r}\mv{x}_i\mv{x}_i^T$$ and
$\mv{x}_i^T(\mv{F}_2-\mv{F}_1)\mv{x}_i\ge 0$ for all $i=1,\ldots,r$.
\end{lemma}

Based on the above lemma, let $y_{ij}=\mv{x}_j^T \mv{F}_i\mv{x}_j$,
$i=0,1,2$, and $j=1,\ldots,r$. Now consider the following linear
program
\begin{align}\label{eq:LP1}
\mathop{\mathtt{Min.}}_{t_1,\ldots,t_r} &~~ \sum_{j=1}^{r} y_{0j} t_j \nonumber \\
\mathtt{s.t.} &~~ \sum_{j=1}^r y_{1j} t_j\geq 1, ~ \sum_{j=1}^r y_{2j} t_j\geq 1 \nonumber \\
&~~ t_j\geq 0,\hspace{4pt} j=1,\ldots,r.
\end{align}
We see that for any feasible set of $t_1,\ldots, t_r$ such that all
the inequality constraints are satisfied, $\mv{X}=\sum_{j=1}^r t_j
(\mv{x}_j\mv{x}_j^T)$ is a feasible solution for the SDP problem
(\ref{eq:PowerMin_SDP}). As such, the minimum objective value of the
above linear program is same as that of the SDP problem
(\ref{eq:PowerMin_SDP}), and one such an optimal point is
$t_1=\ldots=t_r=1$ (which corresponds to $\mv{X}=\sum_{j=1}^r t_j
(\mv{x}_j\mv{x}_j^T)=\sum_{j=1}^r \mv{x}_j\mv{x}_j^T=\mv{X}^\star$).
Note that the optimal points may not be unique.

Furthermore, given that $\mv{x}_i^T(\mv{F}_2-\mv{F}_1)\mv{x}_i\ge 0$
for all $i=1,\ldots,r$ from Lemma~\ref{lemma:decomposition}, we have
$y_{2j}\ge y_{1j}$ for all $j$'s. Therefore, $\sum_{j=1}^r y_{1j}
t_j\geq 1$ implies $\sum_{j=1}^r y_{2j} t_j\geq 1$, i.e., the second
inequality constraint in (\ref{eq:LP1}) is redundant. Thus,
(\ref{eq:LP1}) can be recast as
\begin{align}\label{eq:LP2}
\mathop{\mathtt{Min.}}_{t_1,\ldots,t_r} &~~ \sum_{j=1}^{r} y_{0j} t_j \nonumber \\
\mathtt{s.t.} &~~ \sum_{j=1}^r y_{1j} t_j\geq 1 \nonumber \\
&~~ t_j\geq 0,\hspace{4pt} j=1,\ldots,r.
\end{align}

When $\mv{X}^\star$ can be found for the SDP
problem~(\ref{eq:PowerMin_SDP}), it means that the optimal objective
values for both the SDP problem~(\ref{eq:PowerMin_SDP}) and the
linear program problem~(\ref{eq:LP2}) are bounded, which implies
that (\ref{eq:LP2}) must have one basic optimal feasible solution,
at which at least $r$ inequality constraints are active (to define
an optimal vertex point in the feasible region). Since we only have
$r+1$ inequality constraints in (\ref{eq:LP2}), at most one $t_j$ is
positive. Actually, we have exactly one $t_j$ positive; otherwise,
all zero $t_j$'s could not be a feasible solution. At such a basic
optimal feasible solution, if we have $t^\star_k>0$ and $t_j=0$,
$j\neq k$ with $1\le j, k \le r$, we could infer that there exists
an optimal rank-one solution for the SDP
problem~(\ref{eq:PowerMin_SDP}), which could be constructed as
$$\mv{X}^{\star\star}=t_k^\star(\mv{x}_k\mv{x}_k^T).$$

This completes the proof for Theorem~\ref{theorem:rank_one}.

At last, we present a routine to obtain an optimal rank-one solution
for problem (\ref{eq:PowerMin SNR SDR}) from $\mv{X}^\star$ as
follows:
\begin{enumerate}
\item Decompose $\mv{X}^\star$ in reference to $\mv{F}_2-\mv{F}_1$ as in Lemma~\ref{lemma:decomposition} (For detailed procedure, refer to the proof for Lemma 1
in~\cite{Ye_Zhang}).

\item Construct the linear program problem as shown in~(\ref{eq:LP2}), and solve one basic optimal feasible solution. Such an
algorithm could be based on solving $r$ parallel sub-problems, where
at each sub-problem only one $t_j$ is allowed to take non-zero
values. Then the achieved minimum objective values from the $r$
sub-problems are compared to find the global minimum solution.

\item Given the single optimal positive $t_k^\star$, the rank-one optimal solution for both (\ref{eq:PowerMin SNR SDR}) and (\ref{eq:PowerMin_SDP}) is constructed as $\mv{X}^{\star\star}=t_k^\star(\mv{x}_k\mv{x}_k^T)$.
\end{enumerate}

\section{Proof of Lemma \ref{lemma:rate LB MR}}\label{appendix:proof MR sum rate}

Let $a_{\rm MR}=b_{\rm MR}=\nu$ in $\mv{A}_{\rm MR}$ given by
(\ref{eq:MRR MRT}). We can then show the following equalities:
\begin{eqnarray}
|\mv{h}_1^T\mv{A}_{\rm MR}\mv{h}_2|=|\mv{h}_2^T\mv{A}_{\rm MR}\mv{h}_1|=\nu\theta_1\theta_2(1+\rho) \\
\|\mv{A}_{\rm
MR}^H\mv{h}_1^*\|^2=\|\mv{A}_{\rm MR}\mv{h}_1\|^2=\nu^2\theta_1^2\theta_2(1+3\rho) \\
\|\mv{A}_{\rm
MR}^H\mv{h}_2^*\|^2=\|\mv{A}_{\rm MR}\mv{h}_2\|^2=\nu^2\theta_1\theta_2^2(1+3\rho) \\
\mathtt{tr}(\mv{A}_{\rm MR}\mv{A}_{\rm
MR}^H)=2\nu^2\theta_1\theta_2(1+\rho).
\end{eqnarray}
Let the relay transmit power $p_R$ in (\ref{eq:relay power}) be
equal to the maximum value $P_R$. Using the above equalities, from
(\ref{eq:relay power}) it follows that
\begin{equation}\label{eq:nu MR}
\nu^2=\frac{P_R}{\theta_1\theta_2(1+3\rho)(\theta_1p_1+\theta_2p_2)+2\theta_1\theta_2(1+\rho)}.
\end{equation}
Substituting (\ref{eq:nu MR}) into the above equalities, and from
(\ref{eq:rate 21}) and (\ref{eq:rate 12}), the lower bound on the
sum-rate given in (\ref{eq:MR sum rate}) follows.

\section{Proof of Lemma \ref{lemma:rate LB ZF}}\label{appendix:proof ZF sum rate}
Let $a_{\rm ZF}=b_{\rm ZF}=\nu$ in $\mv{A}_{\rm ZF}$ given by
(\ref{eq:ZFR ZFT}). Denote $\mv{H}_{\rm
UL}^{\dag}=[\mv{a}_1,\mv{a}_2]^T$. From (\ref{eq:rate 21}) and
(\ref{eq:rate 12}), we can show that
\begin{eqnarray}
R_{\rm LB}^{\rm ZF} &\geq&
\frac{1}{2}\log_2\left(1+\frac{\nu^2p_2}{\|\mv{a}_2\|^2\nu^2+1}\right)
\nonumber \\ && +
\frac{1}{2}\log_2\left(1+\frac{\nu^2p_1}{\|\mv{a}_1\|^2\nu^2+1}\right)
\label{eq:inequality ZF 0}
\\ &\geq& \log_2\left(\frac{2p_1p_2}{\|\mv{a}_2\|^2p_1+\|\mv{a}_1\|^2p_2+\frac{p_1+p_2}{\nu^2}}
\right) \label{eq:inequality ZF 1}
\end{eqnarray}
where (\ref{eq:inequality ZF 1}) is due to the Jensen's inequality
(see, e.g., \cite{Cover}) and the convexity of the function
$f(x)=\log_2(1+1/x), x\geq 0$ \cite{Boydbook}. Let the relay
transmit power $p_R$ in (\ref{eq:relay power}) be equal to the
maximum value $P_R$. Then, we obtain from (\ref{eq:relay power})
that
\begin{align}\label{eq:nu ZF}
\nu^2&=\frac{P_R}{\|\mv{a}_2\|^2p_1+\|\mv{a}_1\|^2p_2+\mathtt{tr}\left(\mv{H}_{\rm
UL}^{\dag}(\mv{H}_{\rm UL}^{\dag})^H(\mv{H}_{\rm
DL}^{\dag})^H\mv{H}_{\rm DL}^{\dag}\right)} \\ & \geq
\frac{P_R}{\|\mv{a}_2\|^2p_1+\|\mv{a}_1\|^2p_2+(\|\mv{a}_1\|^2+\|\mv{a}_2\|^2)^2}
\label{eq:inequality ZF 2}
\end{align}
where (\ref{eq:inequality ZF 2}) is due to the fact that
$\mathtt{tr}(\mv{XY})\leq\mathtt{tr}(\mv{X})\mathtt{tr}(\mv{Y})$, if
$\mv{X}\succeq 0$ and $\mv{Y}\succeq 0$ \cite{Horn}. Using
(\ref{eq:inequality ZF 2}), the term inside $\log_2(\cdot)$ in
(\ref{eq:inequality ZF 1}) can be further lower-bounded by
\begin{eqnarray*}
\frac{2p_1p_2}{\left(1+\frac{p_1}{P_R}+\frac{p_2}{P_R}\right)
(\|\mv{a}_2\|^2p_1+\|\mv{a}_1\|^2p_2)+\frac{p_1+p_2}{P_R}(\|\mv{a}_1\|^2+\|\mv{a}_2\|^2)^2}.
\label{eq:inequality ZF 3}
\end{eqnarray*}
Since it can be shown that
$\|\mv{a}_1\|^2+\|\mv{a}_2\|^2=\frac{\theta_1+\theta_2}{\theta_1\theta_2(1-\rho)}$,
substituting this equality into the above equation yields the lower
bound on the sum-rate given in (\ref{eq:ZF sum rate}).

\end{document}